\documentclass[prx,amsmath,twocolumn,showpacs,superscriptaddress]{revtex4-1}

\usepackage{graphicx}
\usepackage{color}
\usepackage{latexsym}
\usepackage{amsmath}
\usepackage{amssymb}
\usepackage{amsthm}
\usepackage{longtable}
\usepackage{subfigure}
\usepackage{mdframed}

\mdfdefinestyle{runexam}{%
    frametitle={Example.},
     linecolor=white,
    innertopmargin=4pt,
    innerbottommargin=8pt,
    innerrightmargin=4pt,
    innerleftmargin=4pt,
        leftmargin = 4pt,
        rightmargin = 4pt,
    backgroundcolor=gray!10!white}

\newcommand{\jao}[1] {}
\renewcommand{\jao}[1] {{\color{OrangeRed}{\bf{JAO: #1}}}}

\newcommand{\jmh}[1] {}
\renewcommand{\jmh}[1] {{\color{blue}{\bf{JMH: #1}}}}

\newcommand{\trg}[1] {}
\renewcommand{\trg}[1] {{\color{ForestGreen}{\bf{TRG: #1}}}}

\RequirePackage[dvipsnames,usenames]{xcolor}

\usepackage{ upgreek }
\usepackage{chemfig}
\usepackage[customcolors,norndcorners]{hf-tikz}
\usepackage{mathtools}

\newcommand{\comments}[1]{}

\newtheorem{proposition}{Proposition}
\newtheorem{corollary}{Corollary}
\newtheorem{lemma}{Lemma}
\newtheorem{theorem}{Theorem}

\newtheorem*{theorem*}{Theorem}
\RequirePackage[dvipsnames,usenames]{xcolor}

\begin{document}

\title{Universal thermodynamic bounds on nonequilibrium response \\ with biochemical applications}

\author{Jeremy A. Owen}
\affiliation{Physics of Living Systems Group, Department of Physics, Massachusetts Institute of Technology, 400 Technology Square, Cambridge, MA 02139}
\author{Todd R. Gingrich}
\affiliation{Department of Chemistry, Northwestern University, Evanston, IL 60208}
\author{Jordan M. Horowitz}
\affiliation{Department of Biophysics, University of Michigan, Ann Arbor, Michigan, 48109,
USA}
\affiliation{Center for the Study of Complex Systems, University of Michigan, Ann Arbor, Michigan 48104, USA}
\affiliation{Department of Physics, University of Michigan, Ann Arbor, Michigan 48109, USA}

\date{\today}

\begin{abstract}
Diverse physical systems are characterized by their response to small perturbations.
Near thermodynamic equilibrium, the fluctuation-dissipation theorem provides a powerful theoretical and experimental tool to determine the nature of response by observing spontaneous equilibrium fluctuations.
In this spirit, we derive here a collection of equalities and inequalities valid arbitrarily far from equilibrium that constrain the response of nonequilibrium steady states in terms of the strength of nonequilibrium driving.
Our work opens new avenues for characterizing nonequilibrium response.
As illustrations, we show how our results rationalize the energetic requirements of two common biochemical motifs. 
\end{abstract}

\maketitle

\section{Introduction}
One of the basic characteristics of any physical system is its response to small perturbations~\cite{Chaiken}.
For instance, response is used to quantify everything from material properties---such as conductivity~\cite{Kubo} and viscoelasticity~\cite{Mason1995}---to the sensing capability of cells~\cite{Qian2012,govern2014PRL} and the accuracy of biomolecular processes~\cite{Bintu2005,Murugan2011,Estrada2016}.
Near thermodynamic equilibrium, response is completely determined by the nature of spontaneous  fluctuations, according to the fluctuation-dissipation theorem (FDT)~\cite{Kubo}.
This deep connection between response and fluctuations is not only of theoretical interest, but also finds practical application.
The FDT forms the basis of powerful experimental techniques, such as microrheology, spectroscopy, and dynamic light scattering~\cite{Chaiken}.
It additionally has implications for the design of mesoscopic devices: highly-responsive equilibrium devices are always plagued by noise. To combine low fluctuations with high sensitivity, a device must be driven away from equilibrium.

The great utility of the FDT near equilibrium~\cite{Chaiken} has led to significant interest in expanding its validity and developing generalizations for nonequilibrium situations. Generically, response can be related to some formal nonequilibrium correlation functions~\cite{Agarwal1972,Risken,Baiesi2009,Prost2009, Seifert2010}.
While these predictions offer fundamental theoretical insight, the necessary correlations are often prohibitively difficult to measure except in simple single-particle systems~\cite{GomezSolano2009,Mehl2010, Bohec2013}. 
In certain special cases, such as under stalling conditions, the FDT holds unmodified~\cite{Altaner2016}. More commonly, however, the study of nonequilibrium response has focused on how the FDT is violated.
For example, the violation of the velocity FDT for Brownian particles can be related to the steady-state heat dissipation through the Harada-Sasa equality \cite{harada2005equality, harada2006energy}; a useful prediction that has been utilized to measure dissipation and efficiency in molecular motors~\cite{Toyabe2010}.
Alternatively, violations of the FDT can be used to fit model parameters, as has been suggested for models of biomolecular processes~\cite{Yan2013,Sato2003}. 
More often FDT violations are framed in terms of system-specific ``effective temperatures''~\cite{Cugliandolo2011, Benisaac2011, Dieterich2015}, whose time-dependence under some circumstances can reveal information about effective equilibrium descriptions.

Inspired by the recent demonstration of thermodynamic bounds on far-from-equilibrium dynamical fluctuations~\cite{Barato2015, Gingrich2016}, we show here that generic nonequilibrium steady-state response can be constrained in terms of experimentally-accessible thermodynamic quantities.
In particular, we present equalities and inequalities---akin to the FDT but valid arbitrarily far from equilibrium---that link static response to the strength of nonequilibrium driving.
Our results open new possibilities to experimentally characterize away-from-equilibrium response and suggest design principles for high-sensitivity, low-noise devices.
As illustrations, we show how our results rationalize the energetic requirements of biochemical switches, biochemical sensing, and kinetic proofreading.

\section{Modeling nonequilibrium steady states}

Nonequilibrium steady states are characterized by the constant and irreversible exchange of energy and matter between a system and its environment.
These flows are driven by thermodynamic affinities---quantities like temperature gradients, chemical potential differences and nonconservative mechanical forces. 
The underlying dynamics leading to the establishment of such steady states are often well-modeled as a continuous-time Markov jump process on a finite set of states $i=1,\dots,N$, which represent (coarse-grained) physical configurations.
The probability $p_i(t)$ to find the system in state $i$ at time $t$ then evolves according to the master equation
\begin{equation}\label{master}
\dot{p}_i(t) = \sum_{j=1}^N W_{ij} p_j(t),
\end{equation}
where the off-diagonal entries of the transition rate matrix $W_{ij}$ specify the probability per unit time to jump from $j$ to $i$, and diagonal entries $W_{ii}=-\sum_{j\neq i}W_{ji}$ are fixed by the conservation of probability.
Time-reversibility of the underlying microscopic dynamics implies that $W_{ij}\neq 0$ only if $W_{ji}\neq 0$~\cite{Seifert2012}. 
We will additionally suppose that for any two states, there is some sequence of allowed transitions ($W_{ij} \neq 0$) connecting them, a property known as irreducibility.
Under this assumption, no matter the initial condition, the solution of the master equation \eqref{master} converges at long times to the unique steady state distribution $\pi$ that satisfies $\sum_{j=1}^NW_{ij}\pi_j=0$.  
This distribution $\pi$, and in particular its dependence on physical quantities through the transition rates $W_{ij}$, serves as a general model of a nonequilibrium steady state.


Many key properties of this nonequilibrium steady state, including its thermodynamics, come to light by picturing the stochastic dynamics described by \eqref{master} playing out on a transition graph---a weighted directed graph $G$, as in Fig.~\ref{fig:graph}$(a)$, where the vertices $\{i\}$ represent the states and directed edges $\{e_{mn}\}$ represent possible transitions, weighted by the rates $W_{mn}$.
\begin{figure}
\includegraphics[scale=.2]{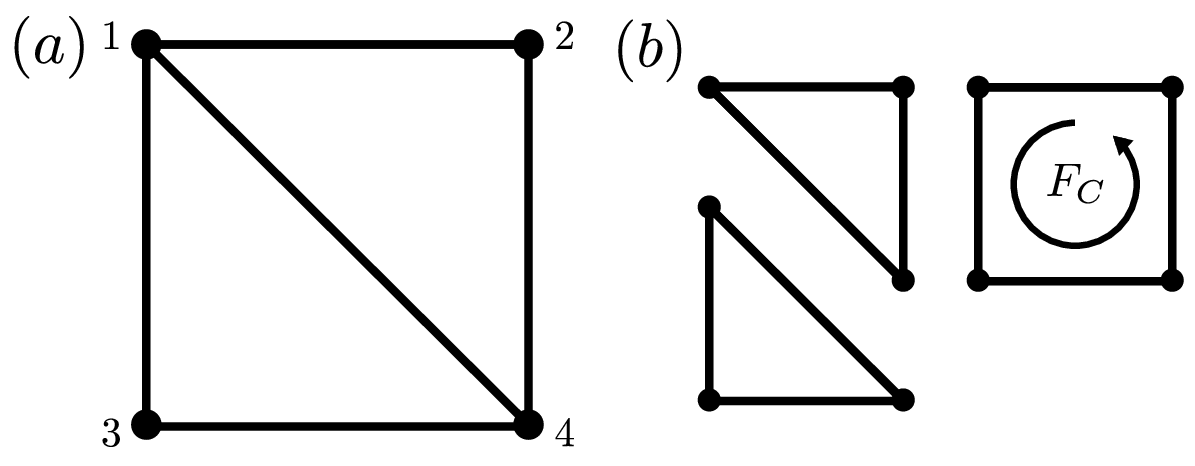}
\caption{{\bf Transition graphs and cycles.} $(a)$ Representative transition graph for a $4$-state system, which will act as a recurring illustrative example. $(b)$ Cycles around which the cycle forces $F_C$ drive the system out of equilibrium.  }\label{fig:graph}
\end{figure}
Note that, by assumption, every edge in $G$ has a reverse, so we will represent and discuss the transition graph as if it were an {undirected} graph, with the understanding that every undirected edge represents two opposing directed edges.

The cycles in the graph, like those in Fig.~\ref{fig:graph}$(b)$, play a central role in the thermodynamics of the steady state.
A {cycle} is a sequence of directed edges and vertices connecting the initial vertex to itself without self-intersecting, $C=\{i_0 \xrightarrow{e_{i_1i_0}} i_1 \to \cdots \to i_{m} \xrightarrow{e_{i_0i_m}}  i_0\}$. 
The asymmetry of the rates around these cycles then encodes the thermodynamic affinities driving the system out of equilibrium through the {cycle forces}---the log of the product of rates around the cycle divided by the product of rates in the reverse orientation~\cite{schnakenberg1976, Andrieux2006}:
\begin{equation}\label{eq:force}
F_{C}=\ln\left(\frac{W_{i_0i_m}\cdots W_{i_2i_1}W_{i_1i_0}}{W_{i_mi_0}\cdots W_{i_1i_2}W_{i_0i_1}}\right).
\end{equation} 
These cycle forces are linear combinations of thermodynamic affinities multiplied by their conjugate distances---for example a chemical potential gradient times a change in particle number. As such, the cycle forces equal the dissipation (entropy production) in the environment accrued every time the system flows around the cycle $C$.
This means that the cycle forces depend on macroscopically tunable parameters---such as environmental temperature or chemical potential---that characterize how strongly the system is driven away from equilibrium.
If all the cycle forces vanish, the system satisfies detailed balance, a statistical time-reversal symmetry~\cite{Zia2007} characteristic of thermodynamic equilibrium.

\section{Static response to perturbations}
Now, suppose the transition rates $W_{ij}(\lambda)$ depend on a control parameter $\lambda$, which could represent, say, the strength of an applied electric field, a temperature, or even a microscopic kinetic parameter such as a reaction barrier.
 In this work, we focus on the response to static perturbations, that is how steady-state averages $\langle Q\rangle_\pi=\sum_j Q_j \pi_j$ respond to small changes in $\lambda$: 
 \begin{equation}
\partial_\lambda \langle Q\rangle_\pi = \sum_i Q_i \partial_\lambda \pi_i =\sum_i Q_i  \sum_{kl}\frac{\partial W_{kl}}{\partial\lambda}\frac{\partial \pi_i}{\partial W_{kl}}.
\end{equation}

 At thermal equilibrium, the steady state $\pi^{\rm eq}_i\propto e^{-\beta \epsilon_i(\lambda)}$ depends only on the 
underlying (free) energy landscape $\epsilon_i(\lambda)$, irrespective of the precise form of the transition rates, where $\beta=1/k_{\rm B}T$ with $k_{\rm B}$ Boltzmann's constant and $T$ temperature.
This simplifying fact immediately implies the static FDT, which equates the static response to an equilibrium correlation function, 
\begin{equation}\label{eqFDT}
\partial_\lambda \langle Q\rangle_{\rm eq} =\beta C_{\rm eq}(Q,V),
\end{equation}
where $C_{\rm eq}(Q,V)=\langle QV\rangle_{\rm eq}-\langle Q\rangle_{\rm eq}\langle V\rangle_{\rm eq}$ and the subscript ``eq'' emphasizes that averages are taken with respect to the equilibrium distribution~\cite{Kubo}.
Here, $V=-\partial_\lambda\epsilon$ is known as the coordinate conjugate to $\lambda$ and represents the displacement induced by $\lambda$---for example, volume is conjugate to pressure and particle number is conjugate to chemical potential.
The FDT's utility in part stems from the fact that we often know the conjugate coordinate from basic physical reasoning and it is easily measured.

Away from equilibrium, the steady-state distribution generally has a complicated dependence on the rates.
Nevertheless, response can  always be related to a nonequilibrium correlation function~\cite{Agarwal1972,Risken},
\begin{equation}\label{noneqFDT}
\partial_\lambda \langle Q\rangle_\pi = C_\pi(Q,\partial_\lambda \ln \pi),
\end{equation}
but the relevant ``conjugate coordinate'' $\partial_\lambda \ln \pi$ requires knowledge of the parameter-dependence of the full nonequilibrium steady-state distribution, which can be challenging to calculate or measure.
Still, this response formula has been given thermodynamic meaning by relating the nonequilibrium conjugate coordinate to the stochastic entropy production rate~\cite{Seifert2010} as well as  to the time-reversal symmetry properties of the path action~\cite{Baiesi2009}.

\section{Parameterizing Perturbations}

Here, we turn our attention to the variations of the steady-state distribution with the transition rates, $\partial \pi_i/\partial W_{kl}$, constraining them in terms of $\pi$ and the cycle forces $F_C$.
We accomplish this goal by breaking any general perturbation into a linear combination of three special types of perturbations that change rates in a coordinated way.
By focusing on these subclasses of perturbations, we will be able to provide clear and measurable thermodynamic constraints on static response.

To classify perturbations, it will prove fruitful to parameterize the rate matrix as
\begin{equation}\label{eq:rates}
W_{ij}=\exp\left[-( B_{ij}-E_j - F_{ij}/2)\right],
\end{equation}
introducing the vertex parameters $E_j$, (symmetric) edge parameters $B_{ij} = B_{ji}$, and asymmetric edge parameters $F_{ij}=-F_{ji}$, which can all be varied independently.
Any rate matrix can be cast in this form, albeit non-uniquely.
To see this, consider the following program for identifying such a parameterization: choose the vertex parameters $\{E_1,\dots,E_N\}$ arbitrarily, then set
\begin{align}
B_{ij} = \frac{1}{2}\left[(E_j-\ln W_{ij})+(E_i-\ln W_{ji})\right]\\
F_{ij}=-\left[(E_j-\ln W_{ij})-(E_i-\ln W_{ji})\right].
\end{align}
The non-uniqueness of the parameterization is manifest in this construction because of the freedom to choose the $E_i$. For example, we could choose $E_i = 0$ for all $i$. We emphasize, however, that no matter the choice, the derivatives with respect to $E_j$, $B_{ij}$ or $F_{ij}$ are independent of the values of the parameterization.
Variations of a vertex parameter, say $E_j$, is equivalent to scaling all of the rates out of state $j$
\begin{equation} \label{eq:vertchainRule}
\frac{\partial \pi_k}{\partial E_j} = \sum_i W_{ij} \frac{\partial \pi_k}{\partial W_{ij}}.
\end{equation}
Similarly, derivatives with respect to $B_{ij}$ and $F_{ij}$ multiplicatively scale transition rates associated with a single edge
\begin{equation}\label{eq:barchainRule}
\frac{\partial \pi_k}{\partial B_{ij}} = -W_{ij} \frac{\partial \pi_k}{\partial W_{ij}} - W_{ji} \frac{\partial \pi_k}{\partial W_{ji}}
\end{equation}
\begin{equation}\label{eq:FchainRule}
\frac{\partial \pi_k}{\partial F_{ij}} = \frac{1}{2}\left(W_{ij} \frac{\partial \pi_k}{\partial W_{ij}} - W_{ji} \frac{\partial \pi_k}{\partial W_{ji}}\right).
\end{equation}
These facts are illustrated in Fig.~\ref{fig:derivative}. We note that the right hand sides of equations \eqref{eq:vertchainRule}--\eqref{eq:FchainRule} do not depend on the choice of parameterization \eqref{eq:rates}, even though the parameterization is not unique.

Our parameterization is reminiscent of the Arrhenius expression for transition rates for a system evolving in an energy landscape with wells of depth $E_i$ and barriers of height $B_{ij}$ driven by forces $F_{ij}$. 
While we stress that \eqref{eq:rates} will not in general support such an interpretation, the analogy is suggestive in several ways.
For example, the asymmetric edge parameters $F_{ij}$ are the sole contributors to the cycle forces (affinities) $F_C=\sum_{e_{ij}\in C}F_{ij}$. Furthermore, if all the $F_{ij}=0$, the steady-state distribution has the Boltzmann form $\pi_i \propto \exp(-E_i)$, with the $E_i$ acting as a dimensionless energy.



\begin{figure}
\centering
\includegraphics[scale=.15]{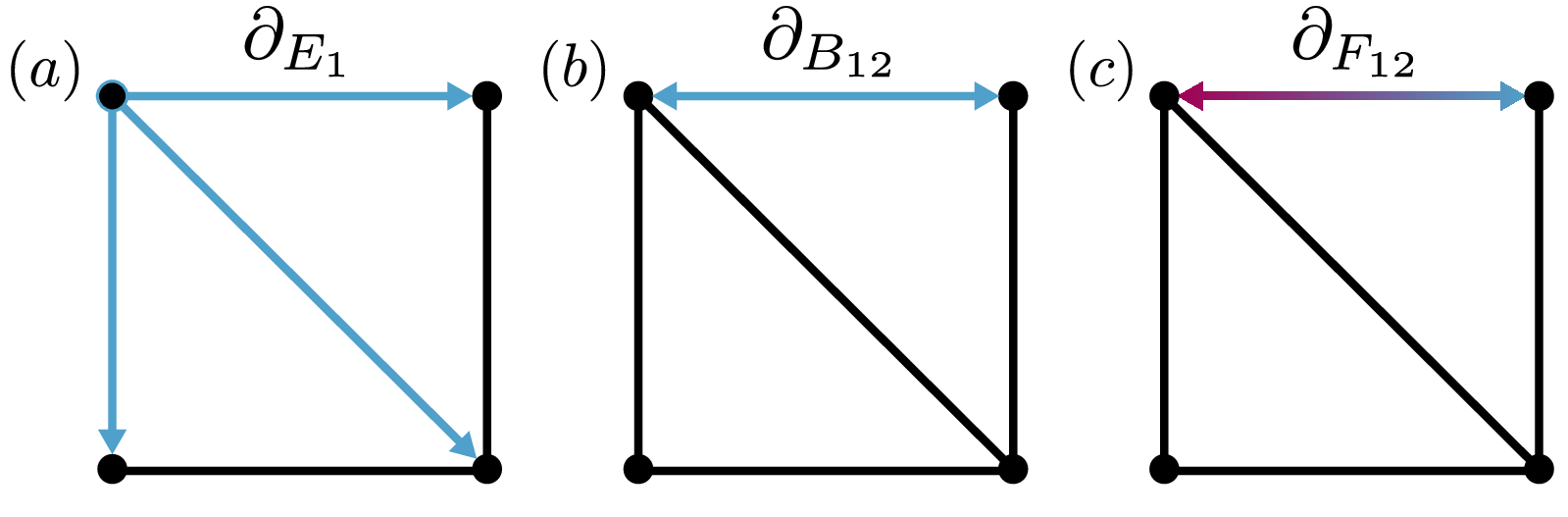}
\caption{{\bf Parameterizing perturbations.} Any perturbation of rates can be decomposed into the variation of some combination of $(a)$ vertex parameters, $(b)$ edge parameters, and $(c)$ asymmetric edge parameters.  Affected rates are highlighted.}\label{fig:derivative}
\end{figure}


Our main results are a series of simple thermodynamic equalities and inequalities for how the steady state responds to perturbations of the $E_j$, $B_{ij}$ and $F_{ij}$.
By combining these results, we can constrain the response to any arbitrary perturbation of the rates through a decomposition of the form
 \begin{equation}
\frac{\partial \pi_k}{\partial \lambda}=\sum_j \frac{\partial E_j}{\partial \lambda}\frac{\partial \pi_k}{\partial E_j}+\sum_{i>j}\frac{\partial B_{ij}}{\partial \lambda}\frac{\partial \pi_k}{\partial B_{ij}}+\sum_{i>j}\frac{\partial F_{ij}}{\partial \lambda}\frac{\partial \pi_k}{\partial F_{ij}}
\end{equation}
The freedom in the rate parameterization makes this decomposition non-unique, and the tightness of our inequalities will depend on the specific decomposition.
For example, one choice could force $\partial_\lambda E_i=0$ for all $i$.
We are not presently aware of a good strategy to identify the decomposition which yields optimally tight inequalities.
In this work, we show that simple decompositions can nevertheless yield interesting bounds.

In deriving our response results, the basic mathematical tool we rely on is the matrix-tree theorem (MTT), presented in Appendix \ref{sec:mtt}, which gives an exact algebraic expression for the steady-state probabilities $\pi_i$ in terms of the rates $W_{ij}$~\cite{schnakenberg1976,Hill}. 
All our results---presented in the following sections---are obtained by reasoning about the result of differentiating the expression given by the MTT. Proofs are given in the appendices.

%
%
%
%
%
%
%

\section{Vertex perturbations}
Our first main result is the exact expression for the response to a vertex perturbation (Appendix~\ref{sec:vertex})
\begin{equation}
\label{eq:ener}
\frac{\partial \pi_i}{\partial E_k} = \begin{cases} -\pi_i(1-\pi_i) &\mbox{if } i =k \\ 
\pi_i \pi_k & \mbox{if } i \neq k \end{cases}.
\end{equation}
We stress that the $B_{ij}$ and $F_{ij}$ are unrestricted, so this equality holds even for nonequilibrium steady states.
As an immediate consequence, we also find that
for $i \neq j$,
\begin{equation}
\frac{\partial \ln \left(\pi_i / \pi_j\right)}{\partial E_k} = \begin{cases} -1 &\mbox{if } i = k \\ 1 &\mbox{if } j = k  \\
0 & \text{otherwise} \end{cases},
\end{equation}
which implies that the relative probability between two states is insensitive to the vertex parameters elsewhere in the graph.

Remarkably, these equalities are exactly equivalent to the response of a Boltzmann distribution to energy perturbations, which leads to the surprising conclusion that far-from-equilibrium response has an equilibrium-like structure if the perturbation leaves the $B_{ij}$ and  $F_{ij}$ fixed.
To leverage this observation, let us assume that we vary the rates only through the system's energy function $\epsilon_i(\lambda)$ and that the rates depend on the energy as $W_{ij}=\omega_{ij}e^{\beta \epsilon_j}$, with arbitrary energy-independent $\omega_{ij}$.
Comparison with \eqref{eq:rates}, shows that variations in the energy $\epsilon_i$  in this case can be parameterized as vertex parameters $E_i$.
Then \eqref{eq:ener} implies that  arbitrarily far from equilibrium the response maintains the equilibrium-like form of the FDT, 
\begin{equation}
\partial_\lambda \langle Q\rangle_\pi = \beta C_\pi(Q,V), 
\end{equation}
with the response proportional to the nonequilibrium steady-state correlation with the coordinate conjugate to the energy $V=-\partial_\lambda \epsilon$ [cf.~\eqref{eqFDT} and \eqref{noneqFDT}].
This prediction implies that experimental verification of the static FDT is not sufficient to conclude that a system is in equilibrium.

\section{Symmetric edge perturbations}

More generally, a perturbation will modify not only the vertex parameters $E_i$, but also the edge parameters $B_{ij}$. 
While at equilibrium the steady state is independent of edge parameters $B_{ij}$, this is generically not the case out of equilibrium.
In this section, we demonstrate that in fact response to edge perturbations is constrained by thermodynamic affinities through the cycle forces.
For proofs of the results in this section, see Appendix \ref{sec:barrier}.

\subsection{Single edge}

Our second main result is that the response to the perturbation of a symmetric edge parameter $B_{mn}$, associated to a single edge $e_{mn}$, is constrained by the cycle forces:
\begin{align}
\label{eq:barBound2}
&\left|\frac{\partial \pi_i}{\partial B_{mn}}\right| \le \pi_i(1-\pi_i)\tanh\left(F_{\rm max}/4\right)\\
\label{eq:barBound1}
&\left|\frac{\partial (\pi_i/\pi_j)}{\partial B_{mn}}\right| \le \left(\frac{\pi_i}{\pi_j}\right)\tanh\left(F_{\rm max}/4\right),
\end{align}
where
\begin{equation}
F_{\rm max}=\max_{C\ni e_{mn}}|F_C|
\end{equation}
 is the maximum cycle force over all cycles that contain the (undirected) perturbed edge $e_{mn}$ (illustrated in Fig.~\ref{fig:barrier}).
 If the cycle forces all equal zero---as they must at equilibrium---then $F_{\rm max}=0$, and the response is zero, as expected.  
In addition, only perturbations of an edge contained in a cycle can induce a response: perturbations of edges whose removal would disconnect $G$ therefore cannot alter the steady state.
 Equation~\eqref{eq:barBound2}, furthermore, has the character of the FDT, once we recognize $\pi_i(1-\pi_i)$ as the variance of the occupation fluctuations of state $i$; thus, we see a manifestation of how thermodynamics shapes the interplay between response and fluctuations.
These inequalities, applying to all discrete stochastic dynamics, significantly generalize a bound for two-state systems derived by Hartich et al.~in a model of nonequilibrium sensing~\cite{Hartich2015}.
\begin{figure}
\includegraphics[scale=0.11]{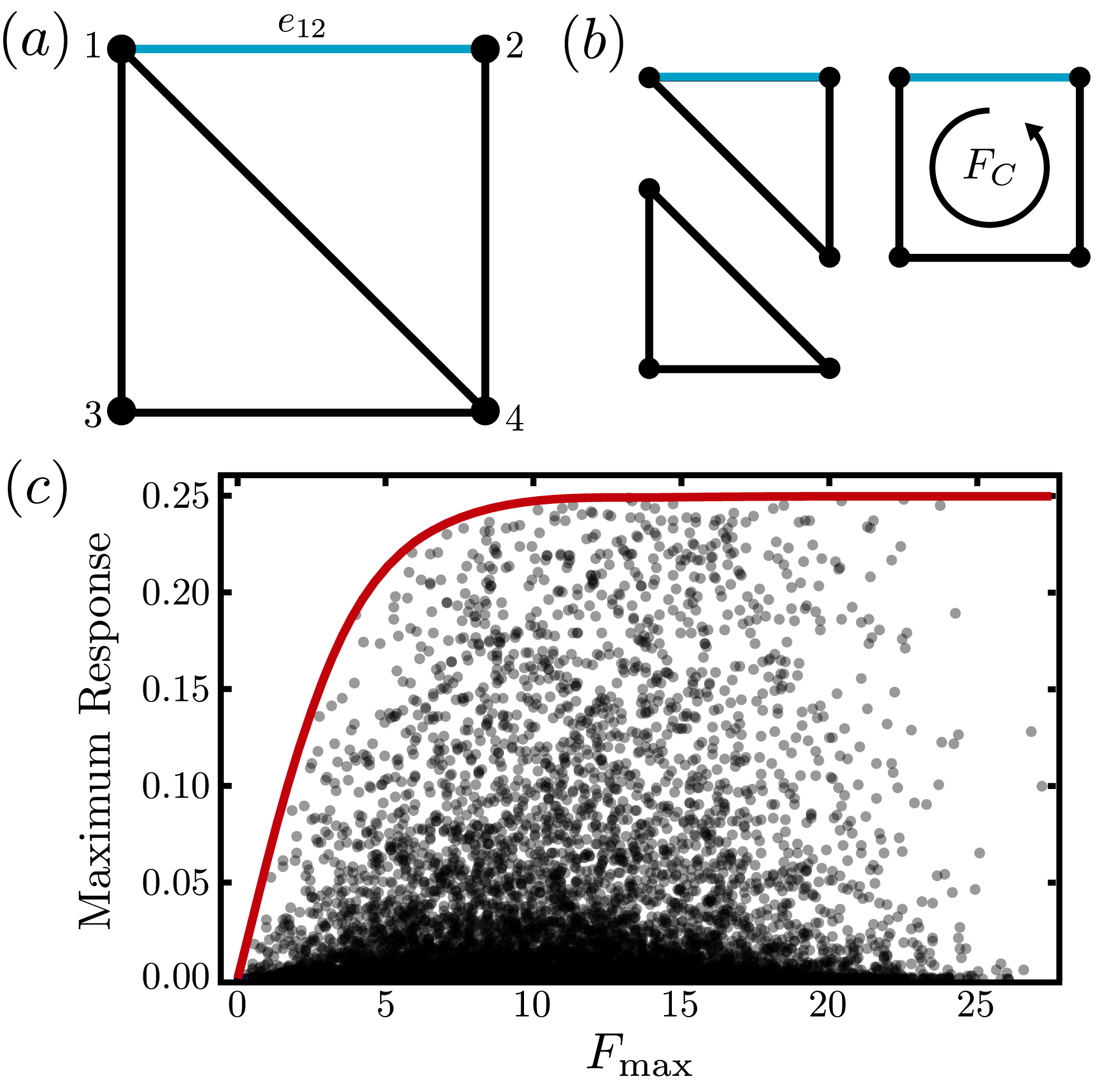}
\caption{{\bf Thermodynamics and topology bound response.} $(a)$ Transition graph for our representative $4$-state system with a single perturbed edge connecting states $1$ and $2$, $e_{12}$, highlighted in blue. $(b)$ Only cycle forces $F_C$ around cycles that contain the perturbed edge, $e_{12}$, highlighted in blue, constrain the static response.  $(c)$ The maximum response $\max_j|\partial \pi_j/\partial B_{12}|$ to the perturbation of the edge parameter $B_{12}$ as a function of maximum cycle force around cycles containing $e_{12}$ for $15000$ randomly sampled rate matrices (grey dots).  All samples fall below the predicted bound $(1/4)\tanh(F_{\rm max}/4)$ (red line).}\label{fig:barrier}
\end{figure}

The conditions for equality in \eqref{eq:barBound2} and \eqref{eq:barBound1} would suggest methods for designing optimized or highly responsive devices.
As detailed in Appendix~\ref{sec:saturate}, we can exhibit at least one scenario that does saturate \eqref{eq:barBound1}: a single cycle with strong time-scale separation so that the system effectively only has two states.
This limiting scenario suggests that small single cycle systems are ideal for optimizing response.
Systematically deducing the system parameters that saturate our inequalities in general remains for future work.

%
%

\subsection{Multiple edges}
 
The response to a perturbation of multiple edge parameters can be bounded by combining \eqref{eq:barBound1} with the triangle inequality. For example, for any set $S$ of $|S|$ edges,
\begin{equation}\label{eq:triangle}
\begin{split}
\left|\sum_{e_{mn} \in S} \frac{\partial \ln (\pi_i / \pi_j)}{\partial B_{mn}}\right|& \leq \sum_{e_{mn} \in S} \left|\frac{\partial \ln (\pi_i / \pi_j)}{\partial B_{mn}}\right| \\
& \leq |S|\tanh(F_\mathrm{max}/4).
\end{split}
\end{equation}
It is clear, however, that this inequality is not always the best we can do. Consider for example the case where $S$ consists of every edge in $G$. In this case, increasing all the edge parameters by the same amount (which is what the sum above amounts to) is like rescaling time, which cannot affect $\pi$ and therefore has zero response.

In this section, we provide a different bound on the response to a perturbation of multiple edge parameters that in many cases improves on \eqref{eq:triangle}. 
Suppose we vary the edge parameters associated to the edges $S$ of a subgraph $H\neq G$. Let $W$ be the set of vertices of $H$ that connect it to the rest of the graph (i.e.~the set of vertices of $H$ incident to an edge not in $H$).
Then,
\begin{equation}\label{eq:barrierBoundMulti}
\left|\sum_{e_{mn} \in S} \frac{\partial}{\partial B_{mn}}\ln \left(\frac{\pi_i}{\pi_j}\right)\right| \leq \left(|W|-1\right)\tanh\left(\frac{F_{i \leftrightarrow j}}{4}\right).
\end{equation}
where $F_{i \leftrightarrow j}$, defined precisely in Appendix \ref{sec:barrier}, can be physically identified as the largest possible entropy produced in the environment when the system goes from $i$ to $j$ and back again (along paths without self-intersection). 
Whenever there is only one path through state space between $i$ and $j$, and in all cases at thermodynamic equilibrium, $F_{i \leftrightarrow j} = 0$.

 
We finally note that our results \eqref{eq:barBound2}, \eqref{eq:barBound1}, and \eqref{eq:barrierBoundMulti} admit generalization to the response of a ratio of positive observables $\langle A\rangle / \langle B\rangle $.  In this general case, the bounds remain unchanged, except that $F_{i \leftrightarrow j}$ is replaced by its maximum over all pairs of vertices $i, j$.

\section{Asymmetric edge perturbations}

Lastly, the MTT allows us to bound the response to asymmetric edge perturbations as (Appendix \ref{sec:force})
\begin{equation}
\label{eq:forceBound}
\left|\frac{\partial \pi_i}{\partial F_{mn}}\right|\le \pi_i(1-\pi_i) \leq \frac{1}{4},
\end{equation}
which is related to, but distinct from, inequalities established in~\cite{thiede2015}. 
This result is a consequence of an identical inequality that holds for general rate perturbations $\partial \pi_i / \partial \ln W_{jk}$.

Any perturbation of rates can be decomposed into a linear combination of perturbations of the vertex and edge parameters $E_i$, $B_{ij}$, and $F_{ij}$ we have introduced, and so the response can be bounded using our inequalities (via the triangle inequality). What our results in this section show is that even for a general perturbations, there is a universal bound---the response to the variation of a single rate is bounded by a constant independent of the structure of $G$ or rates of other transitions, meaning that high sensitivity always requires many different transitions to be perturbed, their cumulative effect generating a response that can greatly exceed $1/4$.  

\section{Biochemical applications}

In this section, we illustrate the use of our main results by detailing applications to well-studied motifs appearing in biochemical networks.

\subsection{Covalent modification cycle}
First, we consider a well-studied model~\cite{Goldbeter1981,Qian2007,xu2012realistic,dasgupta2014fundamental} of a biological switch---the modification/demodification cycle depicted in Fig.~\ref{fig:phospho}, also known as the Goldbeter--Koshland loop \cite{Goldbeter1981, xu2012realistic}, or ``push--pull" network \cite{govern2014PNAS}.

\begin{figure}
	\includegraphics[scale=0.625]{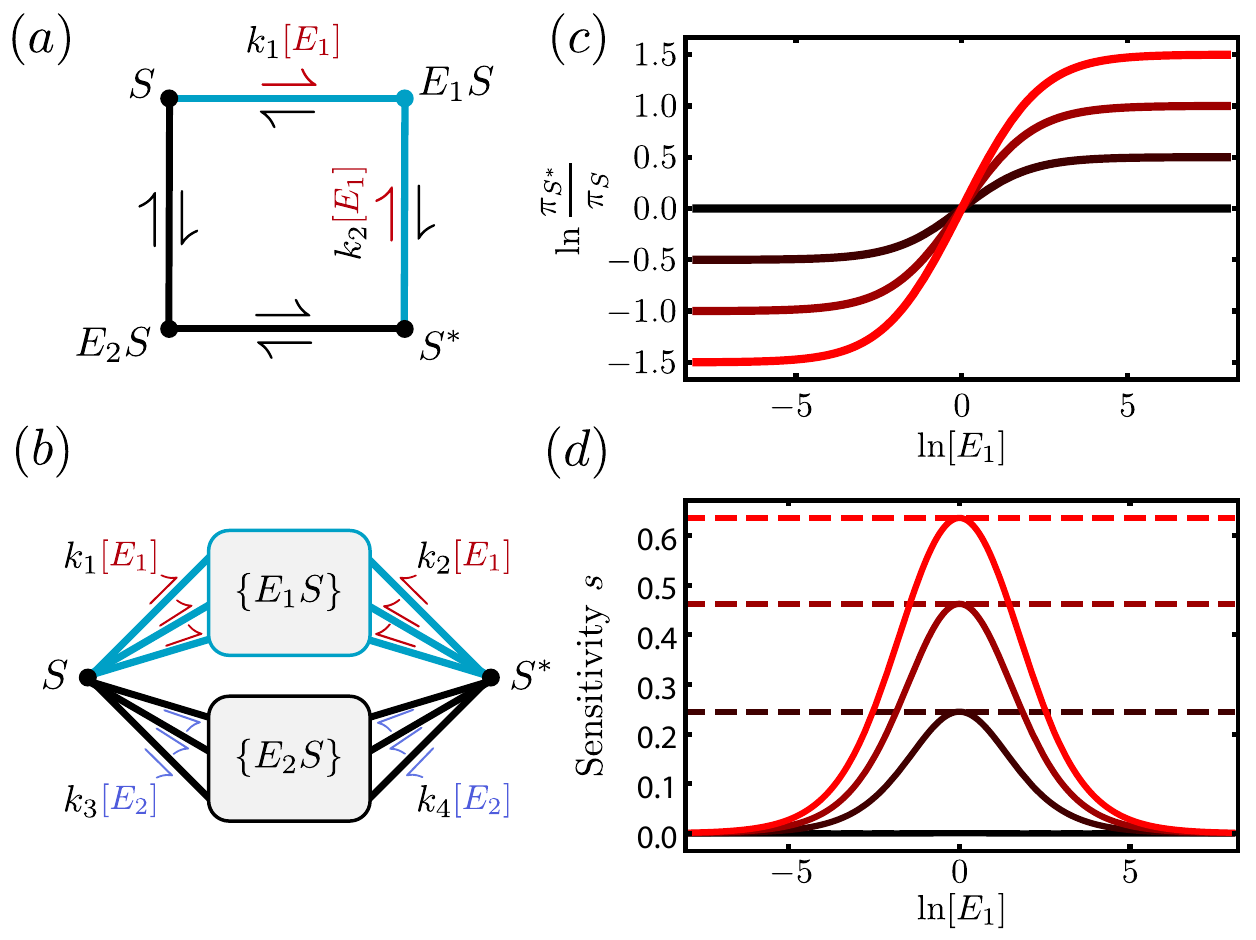}
	\caption{{\bf Sensitivity of a biochemical switch.} Modification/demodification cycle with $(a)$ a single intermediate or $(b)$ general enzymology.   Perturbations of $[E_1]$-dependent rates (red arrows) can be parameterized by blue vertex and edge perturbations.  $(c)$ Single-cycle switch-like behavior of $\ln(\pi_{S^*}/\pi_S)$ as a function of $\ln([E_1])$, with redder curves representing larger chemical driving $\Delta \mu$. $(d)$ Sensitivity $s=\left|\partial \ln(\pi_{S^*}/\pi_S)/ \partial\ln([E_1]) \right|$ as a function of $\ln([E_1])$ for same values of $\Delta \mu$, compared to predicted bound \eqref{eq:switch} (dashed lines).}
	\label{fig:phospho}
\end{figure}
The network consists of a substrate with two forms, an ``unmodified" $S$ and ``modified" $S^*$, along with enzymes, $E_1$ and $E_2$, that actively catalyze its modification and demodification, respectively. 
For example, if  $E_1$ is a kinase, $E_2$ a phosphatase, and $S^*$ a singly-phosphorylated form of $S$, then the system is driven by the chemical potential gradient $\Delta \mu=\mu_{\rm ATP}-\mu_{\rm ADP}-\mu_{{\rm P}_i}$ for ATP hydrolysis.
In the limit in which the substrate is very abundant compared to its modifying enzymes, it is well known that such a system can exhibit unlimited sensitivity to changes in the ratio of the concentrations of the modifying and demodifying enzymes~\cite{Goldbeter1981}.

In the other limit---that of a single substrate molecule---our results \eqref{eq:barBound1} limit the sensitivity of the ratio $\pi_{S^*}/\pi_{S}$ for a particular substrate molecule to changes in the enzyme concentration (Appendix \ref{sec:biochem}):
\begin{equation}\label{eq:switch}
s=\left|\frac{\partial \ln \left(\pi_{S^*}/\pi_S \right)}{\partial \ln [E_1]}\right|\leq \tanh(\Delta \mu/4k_{\rm B}T),
\end{equation}
where $F_{\rm max}=\Delta \mu/k_{\rm B}T$ is the single chemical driving force.
For the simple cycle in Fig.~\ref{fig:phospho}(a) where each enzyme has a single intermediate and we assume mass-action kinetics, this result arises from unraveling a change in $[E_1]$ as a change in the vertex parameter associated to $E_1 S$, together with changes in the parameters of the edges connecting $E_1S$ to $S$ and $E_1S$ to $S^*$.

In fact, inequality \eqref{eq:switch} turns out to hold under assumptions much more general than those of Fig.~\ref{fig:phospho}(a).
It remains true even if catalysis by $E_1$ and $E_2$ proceeds via any number of intermediate complexes with arbitrary rates as in Fig.~\ref{fig:phospho}(b), as long as there is no irreversible formation of a dead-end complex and the chemical driving is the same around every cycle in which $E_1$ makes the modification of $S$ and $E_2$ removes it~\cite{xu2012realistic,dasgupta2014fundamental}.
In this general case, the many perturbed vertex and edge parameters  (Fig.~\ref{fig:phospho}(b)) form a subgraph that acts effectively as a single edge perturbation.
Our multi-edge bound \eqref{eq:barrierBoundMulti} then applies with $F_{S\leftrightarrow S^*}=\Delta\mu/k_{\rm B}T$ being the maximum entropy produced to go from the unmodified $S$ to the modified $S^*$ form and back again.

In the absence of nonequilibrium drive ($\Delta\mu=0$), it is clear this switch cannot work, because it operates by varying the kinetics via an enzyme concentration, and at equilibrium the steady state is independent of kinetics. It has long been known that switches require energy~\cite{Goldbeter1987, Qian2007, tu2008nonequilibrium}. Our results provide a general quantification of this requirement.

\subsection{Biochemical sensing}
The covalent modification cycle, discussed in the previous section, is an integral component of numerous biochemical models for cellular sensing~ \cite{berg1977physics, bialek2005physical, kaizu2014berg, govern2014PRL, govern2014PNAS}.
So far, we have described a single substrate molecule stochastically switching between its two forms due to the action of abundant enzymes $E_1$ and $E_2$.
Here, we imagine there are $N$ substrate molecules, which act as $N$ independent copies of the system studied above, as long as the numbers of both enzymes greatly exceeds $N$. 
Then the number $s$ of modified substrate molecules $S^*$  can  be interpreted as a noisy readout of the enzyme concentration $[E_1]$.
The random variable $s$ is binomially distributed, with mean $\mu_s = N \pi_{S^*}$ and variance $\sigma_{s}^2 = N \pi_{S^*}(1-\pi_{S^*})$, which implicitly depend on $[E_1]$.
Thus, this scenario provides a mechanism to measure  a chemical concentration, by exploiting the relation between $s$ and $[E_1]$.


Now suppose one makes the observation at some time that there are $\tilde{s}$ molecules of $S^*$. Supposing $[E_2]$ and all rate constants assume known, fixed values, one can produce an estimate $\hat{[E_1]}$ of $[E_1]$ by choosing $\hat{[E_1]}$ to be the value of $[E_1]$ that gives $\mu_s(\hat{[E_1]}) = \tilde{s}$. The variance of the estimate $\hat{[E_1]}$ so constructed is often well-approximated, when the noise is small, by~\cite{govern2014PRL, govern2014PNAS} 
\begin{equation}
\label{eq:errorprop}
\sigma_{e}^2 \approx \frac{\sigma_{s}^2}{(\partial \mu_s/\partial [E_1])^2},
\end{equation}
where the quantities on the right hand side should be evaluated at the true concentration $[E_1]$.  Our result \eqref{eq:barBound2} combined with the probabilistic inequality $\pi_{S^*}(1-\pi_{S^*})\le1/4$ then leads to the following bound on the relative error
\begin{equation}
\label{eq:sensingbound}
\left(\frac{\sigma_e}{[E_1]}\right)^2 \gtrsim \frac{4}{N \tanh(\Delta \mu/4 k_\mathrm{B} T)^2}.
\end{equation}
This result interpolates between bounds on error established by Govern and ten Wolde \cite{govern2014PRL, govern2014PNAS} in two limits of resource limitation in cellular sensing systems. That work studies a model of sensing in which cell surface receptors bind to an extracellular ligand whose concentration the cell needs to determine. The ligand-bound receptors then participate in a modification/demodification cycle like the one we study here, playing the role of $E_1$.  See Appendix \ref{tenWolde_consistency} for details.



\subsection{Kinetic proofreading}

As a third application, we turn to the effectiveness of kinetic proofreading~\cite{Hopfield1974, Ninio1975}. A common challenge faced by biomolecular processes is that of discriminating between two very similar chemical species.  At equilibrium, the probability of an enzyme $E$ being bound to a substrate $S$, divided by the probability of that enzyme being free is $\exp(-\Delta)$, where $k_B T \Delta$ is the binding (free) energy of the enzyme-substrate complex $ES$. Two substrates with very similar binding energies are constrained to be bound by the enzyme a similar fraction of time.

\begin{figure}
	\includegraphics[scale=.095]{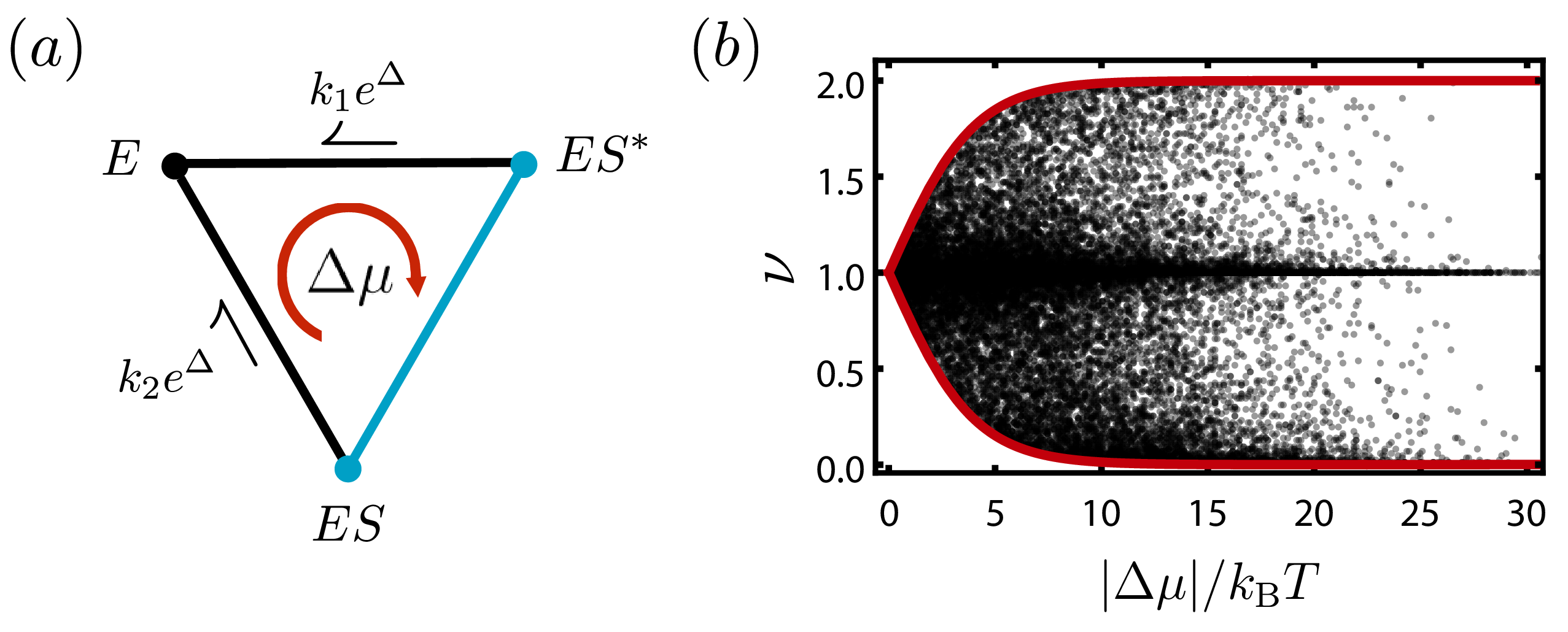}  
\caption{{\bf Bounding the discriminatory index.} $(a)$~Graph of a single-cycle kinetic proofreading network.  Perturbations in the (dimensionless) binding energy $\Delta$ can be unraveled as blue vertex and edge perturbations. $(b)$ Discriminatory index $\nu$ plotted against the thermodynamic affinity $|\Delta \mu| / k_\mathrm{B}T$ for the single-cycle network generated from 30000 randomly sampled transition rates. All samples fall within the predicted bound \eqref{eq:nubound} (red line).}
\label{fig:kp}
\end{figure}
Kinetic proofreading is a scheme to use nonequilibrium driving to improve discrimination based on binding energy. One way to quantify the discriminatory ability of a kinetic network is using the {discriminatory index} introduced by Murugan et al.~\cite{Murugan2014},
\begin{equation}
\nu = - \frac{\partial \ln(\pi_{E}/\pi_{ES})}{\partial \Delta}. 
\end{equation}
At equilibrium, $\nu=1$.
 The simplest nonequilibrium scheme to improve on this is the single-cycle network illustrated in  Fig.~\ref{fig:kp}(a). Note that we have supposed the binding energy $\Delta$ appears exclusively in the unbinding rates. Hopfield observed that in a certain nonequilibrium limit of the rates, $\nu \to 2$~\cite{Hopfield1974}. Our results lead to a constraint on $\nu$ that interpolates between the equilibrium case and this limit. 

In the single-cycle network, the variation of the binding energy $\Delta$ is equivalent to the variation of two vertex parameters (that of $ES$ and $ES^*$) and a single edge parameter ($ES \leftrightarrow ES^*$), leading to the inequality (Appendix \ref{sec:biochem}),
\begin{equation}\label{eq:nubound}
|\nu - 1| \leq \tanh(\Delta\mu/4k_{\rm B}T),
\end{equation}
where $F_\mathrm{max} = \Delta\mu/k_{\rm B}T$ is the chemical driving around the cycle. This bound, which can be saturated, reduces correctly to $\nu = 1$ at equilibrium and is consistent with $\nu \to 2$ in the limit of strong driving $\Delta \mu \to \infty$.

We can also bound $\nu$ in the case of a more general kinetic proofreading scheme \cite{murugan2012speed, Murugan2014} in which there are $m$ complexes that can dissociate. Each of the dissociation transitions can be thought of as crossing a ``discriminatory fence" \cite{Murugan2014}, its rate depending on the binding energy $\Delta$, as in Fig.~\ref{fig:KP_suppl_2}. 
\begin{figure}
	\centering
	\scalebox{0.6}{\includegraphics{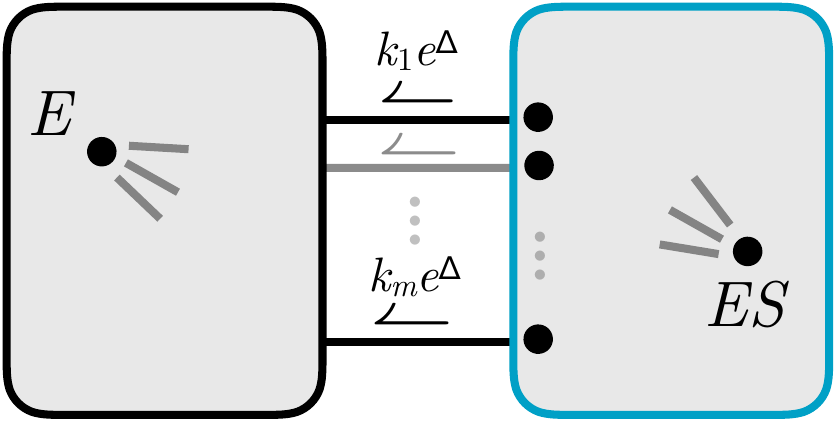}}
	\caption{{\bf Multi-step kinetic proofreading scheme}. The collection of edges with rates that depend on the binding energy $\Delta$ specify a ``discriminatory fence". Perturbing $\Delta$ is equivalent to perturbing vertex and edge parameters of the subgraph labeled in blue.}
	\label{fig:KP_suppl_2}
\end{figure}
We suppose that the dissociation transitions are the only ones that depend on $\Delta$.  We make no assumptions about the structure of the transition graph on either side of the fence. In such a network, perturbing $\Delta$ is equivalent to perturbing the edge and vertex parameters on one side of the ``fence'' forming the subgraph highlighted in blue in Fig.~\ref{fig:KP_suppl_2}.
We then have (Appendix \ref{sec:biochem})
\begin{align}
\left|\nu - 1\right| \le(m-1)\tanh(F_{E \leftrightarrow ES}/4),
\end{align}
where $F_{E \leftrightarrow ES}$ is the maximum entropy produced to go from $E$ to $ES$ and back again.
Notably, we recover the result $\left|\nu - 1\right| \leq m-1$ of \cite{Murugan2014}.

\section{Conclusion}
In this work, we have developed a series of universal bounds on nonequilibrium response in terms of the strength of the nonequilibrium driving. We show that for a large class of static perturbations, a result equivalent to the FDT continues to hold out of equilibrium. For many other perturbations, we bound the response in terms of the dimensionless thermodynamic forces, which quantify departure from equilibrium. 

The illustrations detailed in the previous section demonstrate the potential of our results to unify long-standing observations about the importance of energy ``expenditure" in many different models. The tasks of making a sharp molecular switch, a good sensor, or discriminating between two similar ligands, all have in common the need for a large response to a small perturbation.  We find new bounds interpolating between known limits in these systems, and show how they all descend from our results on vertex and edge perturbations.

A more detailed analysis of the conditions under which our bounds are saturated would lead to design principles for optimal response. 
Our preliminary investigation identified single cycles as ideal when a single edge parameter is varied. We expect that for more complex perturbations, the most highly responsive systems may have a more complicated structure.

An important theme highlighted by our work is that sensitivity is limited not only by nonequilibrium driving, but also, very strictly, by network size and structure. The total number of transitions in a biochemical network limits response, because the response to the scaling of any one rate is bounded by $1/4$. At the same time, our multi-edge results show how many enlargements or complications of networks (e.g.~departure from Michaelis-Menten assumptions in the covalent modification cycle), do not confer any advantage. In this sense, our results build on the work of others who studied similar questions in the context of kinetic proofreading \cite{Murugan2014} and biochemical copy processes \cite{ouldridge2017thermodynamics}.



%
%
%
%

Our results point to numerous other extensions, including bounds on the response of currents with implications for the Green-Kubo and Einstein relations~\cite{Seifert2010b,Baiesi2011,Dechant2018}. 
We have also focused on results that hold in general, not taking into account possible characteristic structures in the graph of states and transitions, which are present for example in many natural examples, such as chemical reaction networks. The study of such extensions and special cases strike us as promising directions for future work. 


%

\section*{Acknowledgments}
We would like to thank Alexandre Solon and Matteo Polettini for very useful discussions. JAO would like to thank Jeremy England for advice and support.

\appendix

\section{Matrix-tree theorem}\label{sec:mtt}

%
The key tool that we apply in our analysis of nonequilibrium response is the matrix-tree theorem (MTT).
%
%
To state the theorem, we must introduce some additional notation and concepts.

For any set of directed edges $S=\{i\to j,k\to l,\dots\}$, we define the weight $w(S)$ to be the \emph{product} of the weights of the edges,
	\begin{equation}
	w(S) = W_{ji}W_{lk}\cdots.
	\end{equation}
	The weight $w(H)$ of any subgraph $H$ we define to be the weight of its edge set.

We also need to introduce {spanning trees}, which are connected subgraphs of a graph $G$ that contain every vertex, but have no cycles, see Fig.~\ref{fig:tree}.
\begin{figure}[h]
\centering
\includegraphics[scale=.75]{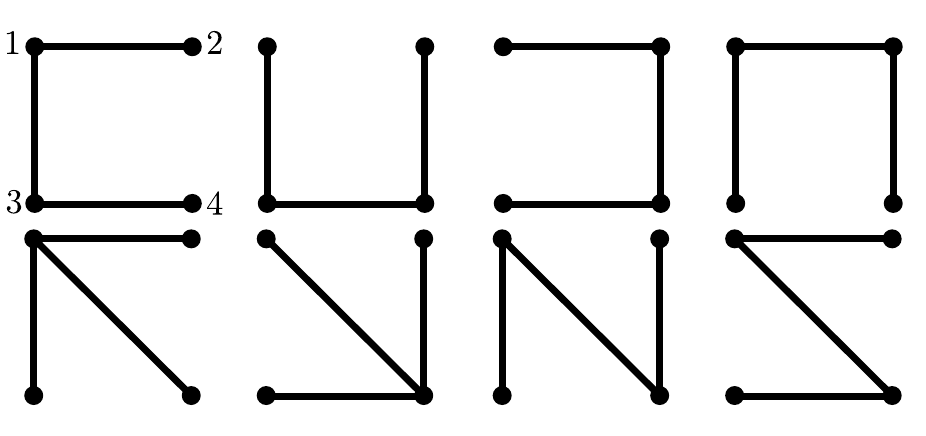}
\caption{{\bf Spanning trees.} All spanning trees for our $4$-state illustrative graph.}\label{fig:tree}
\end{figure}
Every graph that is connected (as is, by assumption, the transition graph of our system) has at least one spanning tree.
For any spanning tree $T$ and vertex $r$ of $G$, there is a unique way to direct the edges of $T$ so that they all ``point towards" $r$, which we then call the ``root". The resulting directed graph, which we write $T_r$, is a {rooted spanning tree} of $G$. 
The steady-state distribution $\pi$ is given explicitly by the matrix-tree theorem (MTT) \cite{tutte1948, hill1966, shubert1975, schnakenberg1976, leighton1986, mirzaev2013} in terms of weights of rooted spanning trees of $G$.

\begin{theorem*}[matrix-tree theorem]\label{mtt}
Let $W$ be the transition rate matrix of an irreducible continuous-time Markov chain with $N$ states. Then the unique steady-state distribution is \begin{equation}
{\pi}_k = \frac{1}{\mathcal{N}}\sum_{\substack{\text{\rm spanning trees} \\ \text{\rm $T$ of $G$}}}\ w(T_k),
\end{equation}
where ${\mathcal N}=\sum_{k=1}^N \sum_T\ w(T_k)$ is the normalization constant.
\end{theorem*}
\noindent This theorem, also known as the Markov chain tree theorem, is a consequence of a result of Tutte \cite{tutte1948}, and has been rediscovered repeatedly in different literatures, see e.g.~\cite{hill1966, shubert1975, schnakenberg1976, leighton1986} and \cite{mirzaev2013} for further discussion.

The MTT offers a graphical representation of the steady-state distribution that provides a convenient method for organizing the structure of the solution.
We illustrate this result in Fig.~\ref{fig:mtt}.
\begin{figure}[h]
\centering
\includegraphics[scale=.6]{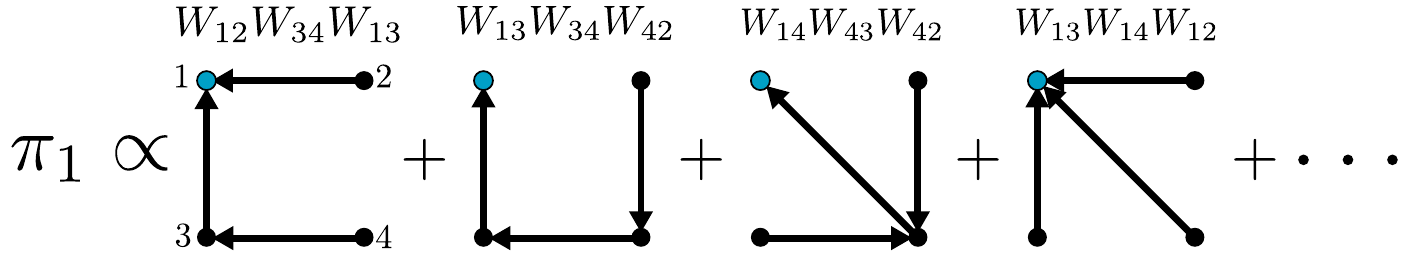}
\caption{{\bf  Matrix-tree theorem.} Graphical representation of the steady-state probability $\pi_1$ as the sum of all spanning trees rooted at $1$ (blue vertex).}\label{fig:mtt}
\end{figure}

\section{Vertex perturbations}\label{sec:vertex}

\begin{theorem}\label{energypert}
\begin{equation}
\frac{\partial \pi_i}{\partial E_k} = \begin{cases} -\pi_k(1-\pi_k) &\mbox{if } i = k \\ 
\pi_k \pi_i & \mbox{if } i \neq k \end{cases}.
\end{equation}
\end{theorem}
\begin{proof}
The matrix-tree theorem implies that $\pi_i$ can be expressed as the ratio of sums of weights of rooted spanning trees. So to evaluate $\partial \pi_i / \partial E_k$, we need to understand in which spanning trees, and in what form, $E_k$ appears. 
The only rates that depend on $E_k$ are rates of transitions {out} of $k$, $W_{*k}=\exp(E_k-B_{*k}+F_{*k}/2)$, see Fig.~\ref{fig:derivative}.  Furthermore, any rooted spanning tree has exactly one edge directed out of $k$, unless the tree is rooted at $k$, in which case it has none. 
These observations allow us to group spanning trees in the MTT expression for the steady-state distribution in a convenient manner as  illustrated in Fig~\ref{fig:vertex}.

Thus, for $i = k$,  the matrix-tree theorem implies that
\begin{equation}
\pi_k = \frac{a}{a + b e^{E_k}} ,
\end{equation}
where
\begin{equation}
a=\sum_{T}w(T_k),\qquad be^{E_k}=\sum_{j\neq k}\sum_{T}w(T_j).
\end{equation}
Here, $a$ is the sum of weights of all spanning trees rooted at $k$---these do not depend on $E_k$ since they have no edge directed out of $k$---and $be^{E_k}$ is the sum of weights of all spanning trees not rooted at $k$---each of these has exactly one factor of $E_k$, making $b$ independent of $E_k$. 

If $i \neq k$, the MTT yields by a similar argument 
\begin{equation}
\pi_i = \frac{c e^{E_k}}{a + b e^{E_k}},
\end{equation}
with 
\begin{equation}
ce^{E_k}=\sum_{T}w(T_i).
\end{equation}
\begin{figure*}[hbt]
\centering
\includegraphics[scale=.65]{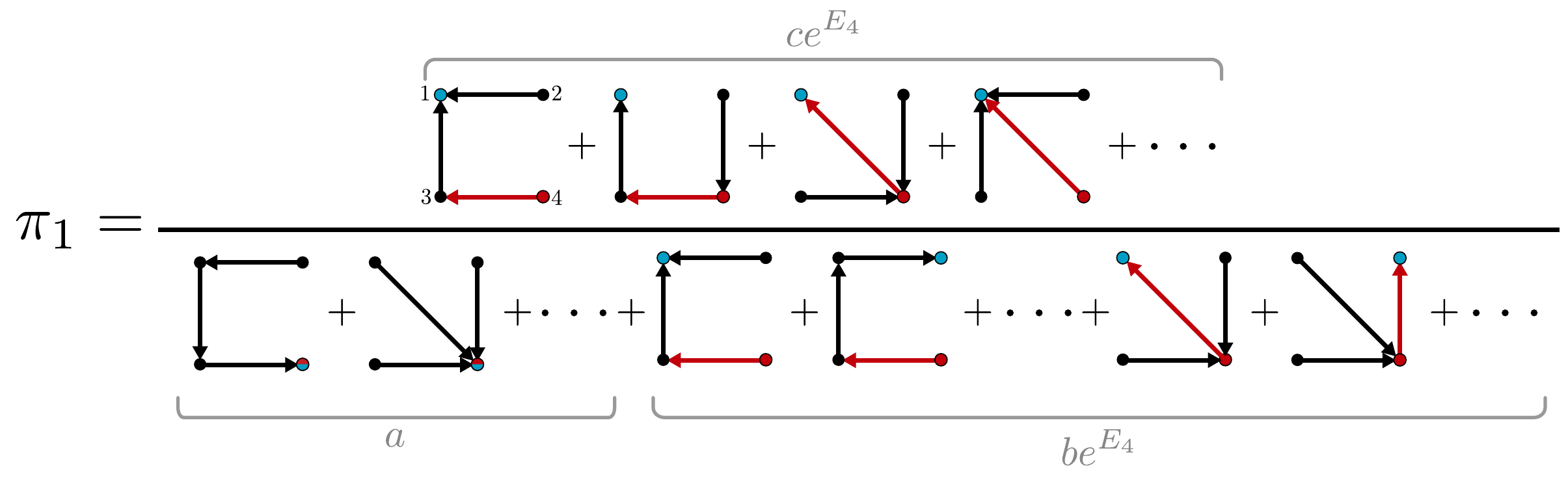}
\caption{{\bf Vertex perturbation.} Groupings of spanning trees with roots labeled in blue for $\pi_1$ ($i=1)$ utilized in the vertex perturbation derivation when the perturbed vertex is $k=4$, labeled in red.  Affected rates labeled in red.}\label{fig:vertex}
\end{figure*}


The theorem now follows by differentiating these expressions.
When $i\neq k$,
\begin{equation}
\begin{split}
\frac{\partial \pi_i }{\partial E_k}& = \frac{c e^{E_k}}{a + b e^{E_k}}-\frac{c e^{E_k}b e^{E_k}}{(a + b e^{E_k})^2} \\
&=\frac{c e^{E_k}}{a + b e^{E_k}}\left(\frac{a}{a + b e^{E_k}}\right)=\pi_i \pi_k,
\end{split}
\end{equation}
and similarly for $i=k$.
\end{proof}

\begin{corollary} 
\label{ratenergypert}
If $i \neq j$,
\begin{equation}
\frac{\partial \ln \left(\pi_i / \pi_j\right)}{\partial E_k} = \begin{cases} -1 &\mbox{if } i = k \\ 1 &\mbox{if } j = k \text{, and}\\ 
0 & \text{otherwise.} \end{cases}
\end{equation}
\end{corollary}
\begin{proof}
First, note that
\begin{equation}
\frac{\partial \ln \left(\pi_i / \pi_j\right)}{\partial E_k} = \frac{1}{\pi_i}\frac{\partial \pi_i}{\partial E_k} -\frac{1}{\pi_j}\frac{\partial \pi_j}{\partial E_k}.
\end{equation}
Now we apply Theorem \ref{energypert}. If $i=k$, then $j \neq k$, and $\frac{\partial \ln \left(\pi_i / \pi_j\right)}{\partial E_k} = -(1-\pi_k)-\pi_k = -1$. If $j=k$, then $i \neq k$, and $\frac{\partial \ln \left(\pi_i / \pi_j\right)}{\partial E_k} = \pi_k+(1-\pi_k )= 1$. And if neither $i$ nor $j$ equal $k$, then $\frac{\partial \ln \left(\pi_i / \pi_j\right)}{\partial E_k} = \pi_k-\pi_k = 0$.
\end{proof}

\section{Symmetric edge perturbations}\label{sec:barrier}


\subsection{Single edge}

In  this appendix, we bound the response to the perturbation of a single symmetric edge parameter in terms of the cycle forces driving the system out of equilibrium. 

First, we prove a general bound on the response of a ratio of observables. Equations \eqref{eq:barBound2} and \eqref{eq:barBound1} will then follow as corollaries by choosing suitable observables.

\begin{theorem}\label{mainresult}
Consider any two observables $A, B \in \mathbb{R}_{\geq 0}^N$ with at least one positive entry. Then,
\begin{equation}
\left|\frac{\partial}{\partial B_{mn}} \ln \frac{\langle A \rangle}{\langle B\rangle}\right| \leq \tanh\left(\frac{F_\mathrm{max}}{4}\right)
\end{equation}
where $F_\mathrm{max}$ is the magnitude of the cycle force that is largest in magnitude, among all those associated to cycles containing the distinguished edge $m \leftrightarrow n$ (in either direction).
\end{theorem}

Our proof relies on the following technical lemma, which we prove in Appendix \ref{sec:rootSwap}.
\begin{lemma}[``Tree surgery"]
	Let ${\mathcal E}_{mn}$ be the set of spanning trees of $G$ containing the distinguished (undirected) edge $m \leftrightarrow n$. Then for any two distinct vertices $i, j$ of $G$,
	\begin{equation} %
	\left|\frac{\sum_{T\in  {\mathcal E}_{mn}}\sum_{S\notin  {\mathcal E}_{mn}}  w(T_i) w(S_j)}{ \sum_{T\in  {\mathcal E}_{mn}}\sum_{S\notin  {\mathcal E}_{mn}} w(T_j) w(S_i) }\right| \leq \exp(F_\mathrm{max}).
	\end{equation}
\end{lemma}

\begin{proof}[Proof of Theorem \ref{mainresult}]
The matrix-tree theorem offers a graphical representation of the steady-state distribution in terms of rooted spanning trees.  
This observation suggests that we can segregate those contributions to steady-state averages that contain $B_{mn}$ by selecting  those (undirected) spanning trees in $G$ that contain the edge $e_{mn}$.  Let us call this set ${\mathcal E}_{mn}$.

Then by the matrix-tree theorem, we can write
	\begin{equation}
	\frac{\langle A\rangle}{\langle B\rangle}=\frac{\sum_i A_i \pi_i}{\sum_j B_j \pi_j} = \frac{a_1 + a_0}{b_1 + b_0}
	\end{equation}
	where 
\begin{align}
	a_1 &= \sum_i \sum_{T \in {\mathcal E}_{mn}} A_i w(T_i) \quad \quad 	a_0 = \sum_i \sum_{S \notin {\mathcal E}_{mn}} A_i w(S_i)\\
	b_1 &= \sum_i \sum_{T \in {\mathcal E}_{mn}} B_i w(T_i) \quad \quad 	b_0 = \sum_i \sum_{S \notin {\mathcal E}_{mn}} B_i w(S_i),
\end{align}
	where  $a_1$ and $b_1$ are linear in $\exp (-B_{mn})$, since they contain edge $e_{mn}$, whereas $a_0$ and $b_0$ are independent of $B_{mn}$.
	An illustrative example is presented in Fig.~\ref{fig:barrierFrac}.
	\begin{figure}
	\centering
	\includegraphics[scale=.48]{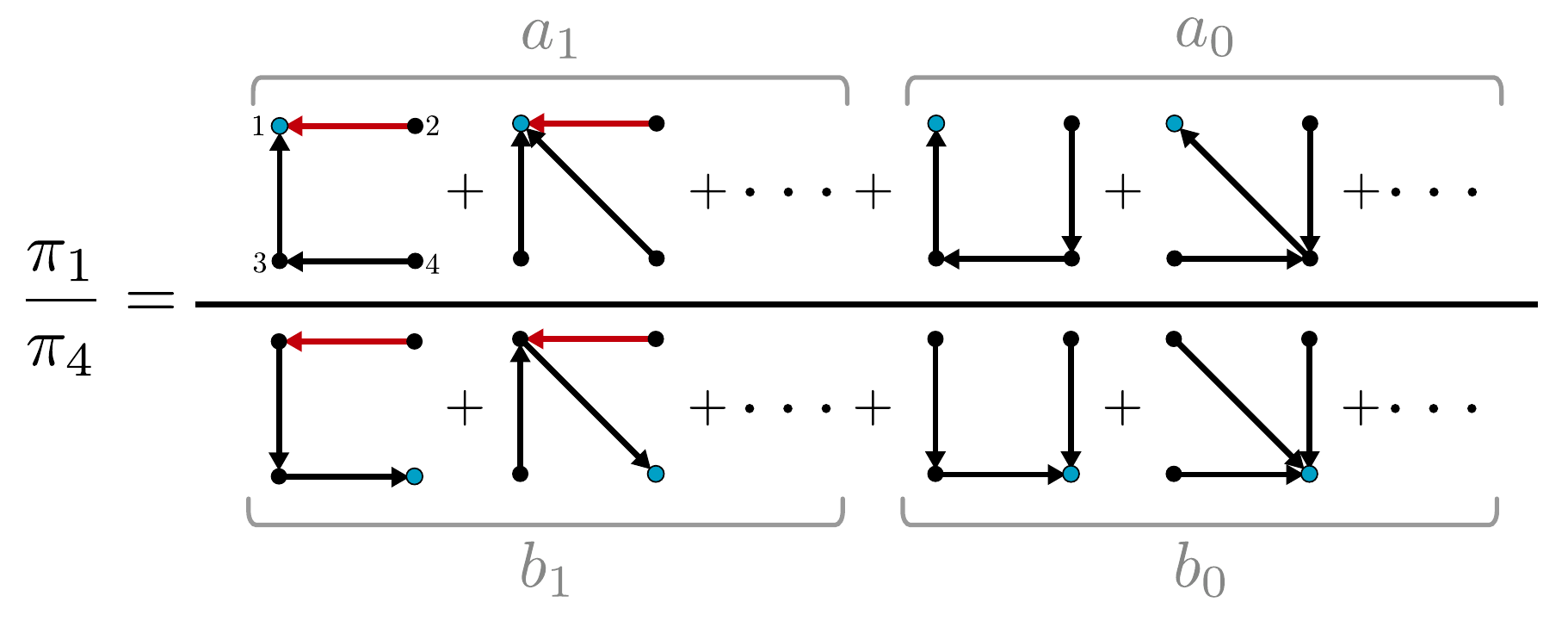}
	\caption{{\bf Symmetric edge perturbation.} Groupings of spanning trees used in the derivation of the symmetric edge perturbation bound for the ratio of observables $\pi_1/\pi_4$ with perturbed edge parameter $B_{12}$.  Roots labeled in blue and affected rates highlighted in red.}\label{fig:barrierFrac}
	\end{figure}
		
	This implies
	\begin{equation}\label{main1}
\frac{\partial}{\partial B_{mn}} \ln \frac{\langle A \rangle}{\langle B\rangle}= \frac{b_0a_1-a_0b_1}{(b_0+b_1)(a_0+a_1)}.
	\end{equation}
	Now note that by the AM-GM inequality the denominator is bounded as
	\begin{equation}\label{main2}
	\begin{split}
	(b_0&+b_1)(a_0+a_1) = b_0 a_0 + b_1 a_0 + b_1 a_1 + b_0 a_1 \\
	&\geq b_0 a_1 + b_1 b_0 + 2\sqrt{a_0 b_0 a_1 b_1} = \left(\sqrt{b_0 a_1} + \sqrt{a_0 b_1}\right)^2.
	\end{split}
	\end{equation}
	Since the numerator $b_0a_1-a_0b_1 = \left(\sqrt{b_0 a_1} - \sqrt{a_0 b_1}\right)\left(\sqrt{b_0 a_1} + \sqrt{a_0 b_1}\right)$, the bound \eqref{main2} implies,
	\begin{equation}\label{main3}
	\left|\frac{\partial}{\partial B_{mn}} \ln \frac{\langle A \rangle}{\langle B\rangle}\right| \leq \left|\frac{\sqrt{b_0 a_1} - \sqrt{a_0 b_1}}{\sqrt{b_0 a_1} + \sqrt{a_0 b_1}}\right|=\tanh\left(\frac{1}{4}\left|\ln\frac{b_0 a_1}{a_0 b_1}\right|\right).
	\end{equation}
	
	To complete the proof, we need to bound the ratio $b_0 a_1 / a_0 b_1$ by $\exp\left(F_{\rm max}\right)$. To do this, we match up terms above and below, writing the fraction as
	\begin{equation}
	\frac{b_0 a_1}{a_0 b_1} = \frac{\sum_i \sum_j \left( A_i B_j \sum_{T\in  {\mathcal E}_{mn}}\sum_{S\notin  {\mathcal E}_{mn}} w(T_i) w(S_j) \right)}{\sum_i \sum_j \left(A_i B_j \sum_{T\in  {\mathcal E}_{mn}}\sum_{S\notin  {\mathcal E}_{mn}} w(S_i) w(T_j)\right)}.
	\end{equation}
	The desired result is now a consequence of the inequality
		\begin{equation}\label{maxfrac}
		\frac{\sum_{i=1}^n x_i}{\sum_{i=1}^n y_i} = \frac{\sum_{i=1}^n \left(x_i/y_i\right)y_i}{\sum_{i=1}^n y_i} \leq \max_i \left(\frac{x_i}{y_i}\right),
		\end{equation}
		to give
		\begin{equation}
		\left|\frac{b_0 a_1}{a_0 b_1}\right| \le \max_{i,j}\left| \frac{\sum_{T\in  {\mathcal E}_{mn}}\sum_{S\notin  {\mathcal E}_{mn}} w(T_i) w(S_j) }{\sum_{T\in  {\mathcal E}_{mn}}\sum_{S\notin  {\mathcal E}_{mn}} w(S_i) w(T_j)}\right|,		
		\end{equation}
	followed by Lemma \ref{treesurgery}.

\end{proof}

From Theorem \ref{mainresult} we readily obtain our bounds on steady-state response.
For \eqref{eq:barBound1}:
\begin{corollary}\label{ratiobound}
\begin{equation}
\left|\frac{\partial \ln \left(\pi_i / \pi_j\right)}{\partial B_{mn}}\right| \leq \tanh\left(\frac{F_\mathrm{max}}{4}\right)
\end{equation}
\begin{proof}
Choose the observables in Theorem \ref{mainresult} to be $A_l=\delta_{il}$ and $B_l=\delta_{kl}$, where $\delta_{ij}$ is the Kronecker delta.
\end{proof}
\end{corollary}

We also have:
\begin{corollary}\label{cor:barboundprobs}
Let be $\pi_X = \sum_{k \in X} \pi_k$ be the total probability of a set of states $X$. Then, 
\begin{equation}
\left|\frac{\partial \pi_X }{\partial B_{mn}}\right| \leq  \pi_X (1-\pi_X) \tanh\left(\frac{F_\mathrm{max}}{4}\right).
\end{equation}
\end{corollary}
\begin{proof}
Choose the observables in Theorem \ref{mainresult} to be $A_i=\delta_i(X)$ and  $B_i=1-\delta_i(X)$, where the indicator $\delta_i(X) = 1$ if $i \in X$ and $\delta_i(X) = 0$ otherwise. Note that we then have $\langle A \rangle = \pi_X$ and $\langle B \rangle  = 1-\pi_X$.
\end{proof}
\noindent If $X=\{i\}$ consists of only a single state we recover the bound \eqref{eq:barBound2}.

\subsection{Multiple edges}

In this section, we derive our inequality for the response to perturbations by multiple edge parameters \eqref{eq:barrierBoundMulti}.
The proof proceeds in two steps.  
We first prove a bound on an arbitrary set of edges $S$ from which \eqref{eq:barrierBoundMulti} and other results are ready corollaries.

%
%
Here, the magnitude of response is bounded by a different function $F_{i \leftrightarrow j}$ of cycle forces. The quantity $F_{i \leftrightarrow j}$ is defined for any graph $G$ and vertices $i$ and $j$ to be
\begin{equation}\label{eq:pathForce}
F_{i \leftrightarrow j} = \max_{P_{i\to j},P_{j\to i}} \left|\ln \frac{w(P_{i\to j} \cup P_{j\to i})}{w({P}^*_{i\to j}\cup {P}^*_{j\to i})}\right|.
\end{equation}
where $P_{i\to j}$ is a (non-self-intersecting) path from $i$ to $j$, $ P_{j\to i}$ is a (non-self-intersecting) path from $j$ to $i$, and the superscript `$*$' denotes the reverse path. 
\begin{theorem}\label{multiedgeCycle}
	Let $S$ be a set of edges, and define $c_\mathrm{max}$ to be the size of the largest intersection $S$ has with any spanning tree of $G$. Similarly, define $c_\mathrm{min}$ to be the size of the smallest such intersection. Then,
	\begin{equation}
	\left|\sum_{e_{mn} \in S} \frac{\partial}{\partial B_{mn}}\ln \left(\frac{\pi_i}{\pi_j}\right)\right| \leq \left(c_\mathrm{max} - c_\mathrm{min}\right)\tanh\left(\frac{F_{i \leftrightarrow j}}{4}\right).
	\end{equation}
\end{theorem}
The appearance of $F_{i \leftrightarrow j}$ in this result stems from this lemma, that we rely on here and prove in Appendix \ref{sec:rootSwap}.
\begin{lemma}[``Cycle flip only"]
	\label{rootswap}
	For any spanning trees $T, S$ and vertices $i, j$ of $G$,
	\begin{equation}
	\frac{w(T_i)w(S_j)}{w(T_j)w(S_i)} \leq \exp(F_{i \leftrightarrow j}).
	\end{equation}
\end{lemma}

We will also rely on the following lemma, which generalizes the first part of the proof of Theorem \ref{mainresult}.
\begin{lemma}\label{multicycleLem}
For any symbols $\{a_n\}, \{b_n\}$,
\begin{equation}\begin{split}
&\left|\frac{\left(\sum_{n=i}^j n a_n\right) \sum_{n=i}^j b_n -\left(\sum_{n=i}^j n b_n\right) \sum_{n=i}^j a_n}{\sum_{n=i}^j a_n \sum_{n=i}^j b_n}\right| \\
&\qquad\qquad\leq \sum_{m={i+1}}^j \tanh\left(\frac{1}{4} \ln \left|\frac{\sum_{n=m}^j a_n \sum_{n=i}^{m-1} b_n}{\sum_{n=m}^j b_n \sum_{n=i}^{m-1} a_n}\right|\right)
\end{split}
\end{equation}
\end{lemma}
\begin{proof}
	
	First note we can rearrange the sum as
	\begin{equation}\label{resum}
	\sum_{n=i}^j n a_n = \sum_{n=i}^j \sum_{m=1}^n a_n=i \sum_{n=i}^j a_n +\sum_{m=i+1}^j \sum_{n=m}^j a_n,
	\end{equation}
	which is illustrated in Fig.~\ref{fig:resum}.
\begin{figure}
\begin{center}
	\scalebox{0.3}{\includegraphics{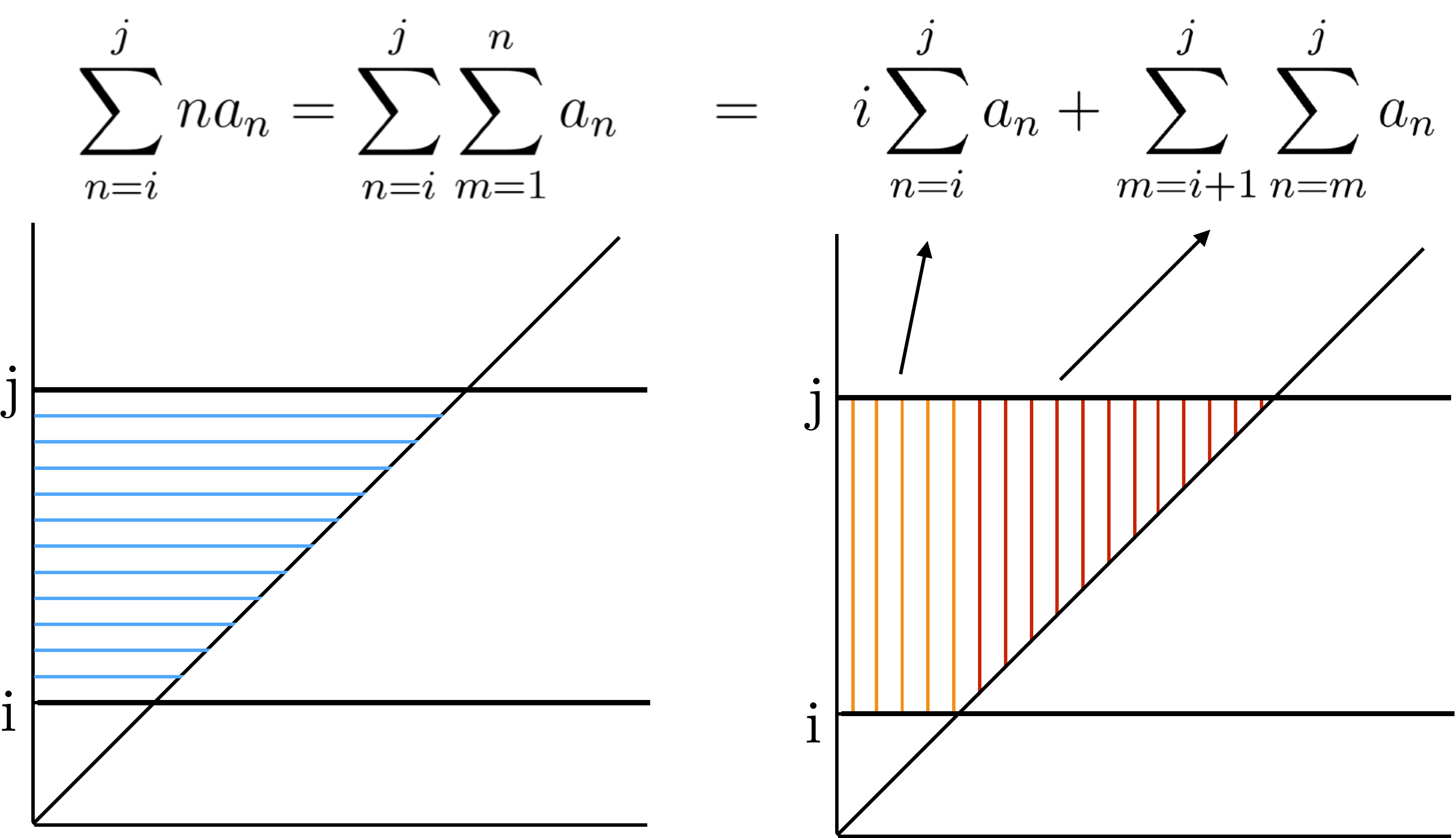}}
\end{center}
	\caption{{\bf Illustration of the identity \eqref{resum}.}}
	\label{fig:resum}
\end{figure}
As a result, we have
\begin{equation}
\begin{split}
&\frac{\left(\sum_{n=i}^j n a_n\right) \sum_{n=i}^j b_n -\left(\sum_{n=i}^j n b_n\right) \sum_{n=i}^j a_n}{\sum_{n=i}^j a_n \sum_{n=i}^j b_n}\\ 
&= \sum_{m=i+1}^j \frac{\left(\sum_{n=m}^j a_n\right) \sum_{n=i}^j b_n -\left(\sum_{n=m}^j b_n\right) \sum_{n=i}^j a_n}{\sum_{n=i}^j a_n \sum_{n=i}^j b_n}\\ 
& =\sum_{m={i+1}}^j \frac{\left(\sum_{n=m}^j a_n\right) \sum_{n=i}^{m-1} b_n -\left(\sum_{n=m}^j b_n\right) \sum_{n=i}^{m-1} a_n}{(\sum_{n=i}^{m-1} a_n+\sum_{n=m}^{j} a_n)( \sum_{n=i}^{m-1} b_n+\sum_{n=m}^j  b_n)}.
\end{split}
\end{equation}
Comparison with equations \eqref{main1} through \eqref{main3} in the proof of Theorem \ref{mainresult}, together with the triangle inequality, establishes the desired result.
\end{proof}

Now we are ready to proceed with the proof of the theorem.
\begin{proof}[Proof of Theorem \ref{multiedgeCycle}]
Define 
\begin{equation}
a_c = \sum_{T : |T \cap S| = c}  w(T_i),\quad
b_c = \sum_{T : |T \cap S| = c}  w(T_j).
\end{equation}
so that we have, for all $c$
\begin{equation}
\sum_{e_{mn} \in S} \frac{\partial a_c}{\partial B_{mn}} = c a_c \, , \quad \sum_{e_{mn} \in S} \frac{\partial b_c}{\partial B_{mn}} = c b_c.
\end{equation} 

By the matrix-tree theorem, the derivative we wish to bound can be written in terms of these quantities as
\begin{equation}
\begin{split}
\sum_{e_{mn} \in S} \frac{\partial}{\partial B_{mn}}\ln \left(\frac{\pi_i}{\pi_j}\right) &= \sum_{e_{mn} \in S} \frac{\partial}{\partial B_{mn}}\ln \frac{\sum_T w(T_i)}{ \sum_T w(T_j)} \\
&= 
 \sum_{e_{mn} \in S} \frac{\partial}{\partial B_{mn}}\ln \frac{\sum_{c=c_\text{min}}^{c_\text{max}} a_c}{\sum_{c=c_\text{min}}^{c_\text{max}} b_c}.
 \end{split}
\end{equation}
Expanding the derivative and applying Lemma \ref{multicycleLem} yields
\begin{equation}
\begin{split}
&\left|\sum_{e_{mn} \in S} \frac{\partial}{\partial B_{mn}}\ln\left(\frac{\pi_i}{\pi_j}\right)\right| \\
&\qquad \leq \sum_{m={c_\mathrm{min}+1}}^{c_\mathrm{max}} \tanh\left(\frac{1}{4} \ln \left|\frac{\sum_{n=m}^{c_{\mathrm{max}}} a_n \sum_{n=c_\mathrm{min}}^{m-1} b_n}{\sum_{n=m}^{c_\mathrm{max}} b_n \sum_{n=c_\mathrm{min}}^{m-1} a_n}\right|\right).
\end{split}
\end{equation}

To prove the theorem, all that remains  is to demonstrate that
\begin{equation}
\frac{\sum_{n=m}^{c_{\mathrm{max}}} a_n \sum_{n=c_\mathrm{min}}^{m-1} b_n}{\sum_{n=m}^{c_\mathrm{max}} b_n \sum_{n=c_\mathrm{min}}^{m-1} a_n} \leq \exp(F_{i \leftrightarrow j})
\end{equation}
holds for all $m$. This follows by an application of Lemma \ref{rootswap}. So we have
\begin{equation}
\begin{split}
\left|\sum_{e_{mn} \in S} \frac{\partial}{\partial B_{mn}}\ln \frac{\pi_i}{\pi_j}\right| &\leq \sum_{m={c_\mathrm{min}+1}}^{c_\mathrm{max}} \tanh\left(\frac{F_{i \leftrightarrow j}}{4}\right) \\
&= \left(c_\mathrm{max} - c_\mathrm{min}\right)\tanh\left(\frac{F_{i \leftrightarrow j}}{4}\right)
\end{split}
\end{equation}
as desired.
\end{proof}

Theorem \ref{multiedgeCycle} has a number of simple corollaries.

\begin{corollary}\label{rcycles}
If $G$ has $r$ independent cycles, then for any set $S$ of edges,
	\begin{equation}
	\left|\sum_{e_{mn} \in S} \frac{\partial}{\partial B_{mn}}\ln \left(\frac{\pi_i}{\pi_j}\right)\right| \leq r \tanh\left(\frac{F_{i \leftrightarrow j}}{4}\right).
	\end{equation}
\end{corollary}
\begin{proof}
	Let $m$ be the number of edges in $G$. The largest possible intersection of a spanning tree and $S$ cannot exceed $|S|$ in size, so we have $c_\mathrm{max} \leq |S|$. Furthermore, each spanning tree of $G$ has exactly $m-r$ edges. So the smallest possible intersection is realized if all $r$ edges a spanning tree excludes are edges in the set $S$, which means $c_\mathrm{min} \geq |S|-r$. Therefore, $c_\mathrm{max}- c_\mathrm{min} \leq r$, and the corollary follows from Theorem \ref{multiedgeCycle}. 
\end{proof}

\begin{corollary}\label{forGK}
Let $H$ be a subgraph of $G$, and write $W$ for the set of vertices of $H$ incident to an edge not in $H$. Let $S$ be the edge set of $H$. Then,
\begin{equation}
\left|\sum_{e_{mn} \in S} \frac{\partial}{\partial B_{mn}}\ln \left(\frac{\pi_i}{\pi_j}\right)\right| \leq \left(|W|-1\right)\tanh\left(\frac{F_{i \leftrightarrow j}}{4}\right).
\end{equation}
\end{corollary}

\begin{proof}
Consider a spanning tree $T$ of $G$. Viewed as a subgraph of $H$, $T$ is still at least a spanning forest (i.e.~it may no longer be connected, but still has no cycles and includes every vertex of $H$), with no more than $|W|$ component trees. To see this, suppose it had $|W|+1$ component trees. In this case, one component would have to be disconnected from all the vertices in $W$ (if every component is connected to a vertex in $|W|$, there can be at most $|W|$, as components cannot be connected to each other). But in that case, $T$ (as a subgraph of $G$) was disconnected---it was never a spanning tree at all. 

Let $n$ be the number of vertices in $H$. The number of edges in a spanning forest is always the number of vertices in the forest minus the number of components (trees in the forest). This means that for our graph $G$, the size of the intersection of $S$ and the edge set of $T$ is restricted to lie between $ n - 1 = c_\mathrm{max}$ or $n - |W| = c_\mathrm{min}$. By Theorem \ref{multiedgeCycle}, this implies the result.
\end{proof}


\section{Proofs of the root-swapping lemmas}\label{sec:rootSwap}

In the course of proving our results above we came across ratios of products of spanning tree weights, such as  
\begin{equation}
	\frac{\sum_{T\in  {\mathcal E}_{mn}}\sum_{S\notin  {\mathcal E}_{mn}}  w(T_i) w(S_j)}{ \sum_{T\in  {\mathcal E}_{mn}}\sum_{S\notin  {\mathcal E}_{mn}} w(T_j) w(S_i)},
\end{equation}
	which we bounded using Lemmas \ref{treesurgery} and \ref{rootswap}, yielding our theorems. 
	Here, we present proofs of these key lemmas. 
	The arguments will depend on the existence of invertible mappings between the pairs of spanning trees in the numerator to pairs of spanning trees in the denominator, which have their roots ``swapped": $(T_i,S_j)\to (T^\prime_j,S^\prime_i)$.  
	We will construct these mappings explicitly, but first we set out some relevant notation and definitions.

First, we will find it helpful in this section  to use the standard notation $s(e)$ (the {source}) for the vertex at the tail of a directed edge $e$  and $t(e)$ (the {target}) for the vertex at the head of $e$, where the arrow points.
In addition, the graph formed by the removal of the edge $h$ from a graph $H$, i.e.~by the deletion of $h$, will be denoted $H\setminus h$, and the graph formed by adding an edge $h$ to $H$ will be denoted $H \cup h$.

Second, we need to define a new kind of spanning tree. We have already introduced spanning trees, as well as the notion of a spanning tree $T_i$ rooted at vertex $i$. Recall that in a rooted spanning tree, every edge is directed {towards} the root $i$ (since a tree has no cycles, this direction is defined unambiguously). Generalizing from this, we define a {doubly-rooted spanning tree}, schematically depicted in Fig.~\ref{fig:edgeswap}(a). We start with a spanning tree $S$ and two vertices $i$ and $j$.
We first note that all the edges in the rooted trees $S_i$ and $S_j$ are oriented in the same direction except for those edges along the unique path {between} $i$ and $j$. This inspires us to pick a vertex $k$ on this path and define a doubly-rooted spanning tree $S_{ij,k}$ {with branch point} $k$ to be the spanning tree $S$ with every edge directed as it is in $S_i$ and $S_j$---when those directions are the same---and otherwise directed towards $i$ if between $k$ and $i$, and towards $j$ if between $k$ and $j$. 
One can view a (singly) rooted spanning tree $S_j$ as a sort of ``degenerate" doubly-rooted tree $S_{ij,i}$ with branch point $i$.

Our mappings are then built from repeated applications of the following operations on pairs of the form $(T_b, S_{mn,b})$, where $T_b$ is a spanning tree rooted at some vertex $b$, and $S_{mn,b}$ is a doubly-rooted spanning tree with roots $m$ and $n$ and branch point $b$:
\begin{itemize}

\item \emph{Cycle flip.} Consider the unique edge $e$ pointing out of $b$ towards $n$ in $S_{mn,b}$. Reroot the tree $T_b$ to the target $t(e)$ to form $T_{t(e)}$ and flip the edge $e\to e^*$ to form $S_{mn,t(e)}=\left(S_{mn,b} \setminus e\right)\cup e^*$. Output $(T_{t(e)}, S_{mn,t(e)})$.

\item \emph{Edge swap.} Consider the unique edge $e$ pointing out of $b$ towards $n$ in $S_{mn,b}$. Let $f$ be first edge, along the unique directed path in $T_b$ from $t(e)$ to $b$, that reconnects $S_{mn,b} \setminus e$. Swap these edges to form $T'_{s(f)}=(T_b \setminus f) \cup e$ and $S'_{mn,s(f)}=(S_{mn,b} \setminus e) \cup f$. Output $(T'_{s(f)}, S'_{mn,s(f)})$.
\end{itemize}
\begin{figure}
	\centering
	\includegraphics[width=\textwidth/2]{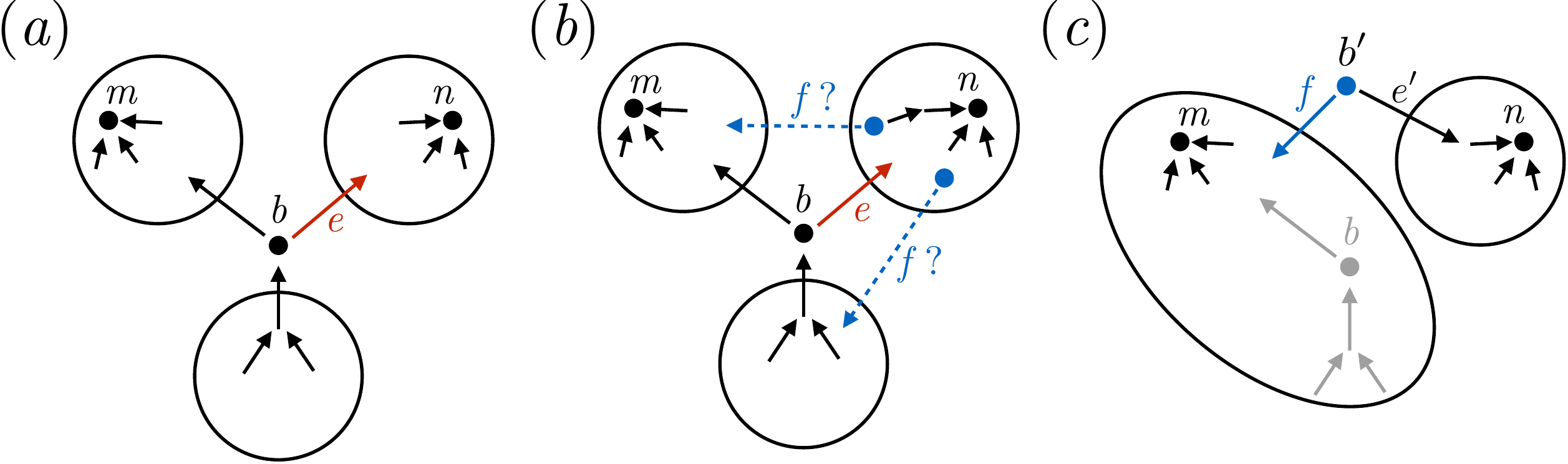}
	\caption{{\bf Effect of edge swap on $S_{mn,b}$.} (a) The structure of $S_{mn,b}$, a doubly-rooted tree with roots $m,n$ and branch point $b$. The edge $e$ is the unique edge pointing out of $s(e)=b$ towards $n$. (b) Removing the edge $e$ disconnects $S_{mn,b}$, since it is a tree. The edge $f$ from $T_b$ that reconnects it can point from the component containing $n$ to any part of the component containing $m$. It is also possible that $f = e^*$. (c) No matter where $f$ points, removing $e$ and adding $f$ to $S_{mn,b}$ yields a new doubly-rooted tree $S'_{mn,b'}$, with branch point $b' = s(f)$.}
	\label{fig:edgeswap}
\end{figure}

The output of each of these operations is another pair $(T'_{b'}, S'_{mn,b'})$ consisting of a tree rooted at $b'$ and a doubly-rooted tree with branch point $b'$ (see Figure \ref{fig:edgeswap} for an illustration of this in the case of \emph{edge swap}). 
Furthermore, no edges are reoriented in the \emph{edge swap}, although edges are exchanged between $T$ and $S$.
In a general \emph{cycle flip}, no edges are exchanged, and the edges that are reoriented form the single cycle obtained from the union of $e$ with the unique path in $T$ from $t(e)$ to $s(e)$.
In the degenerate case where the path in $T$ from $t(e)$ to $s(e)$ consists of the single edge $e^*$, \emph{cycle flip} and \emph{edge flip} are equivalent.

Notably, both of these operations are invertible, in the sense that given the output of either, {and} knowledge of which was applied, we can uniquely recover the original pair $(T_b,S_{mn,b})$ from $(T'_{b'}, S'_{mn,b'})$. 

\begin{itemize}
	
	\item To invert the \emph{cycle flip} all we need is to identify the original edge $e$---it is the reverse of the unique edge pointing out of $b'$ towards $m$. Note that $s(e) = b$, the original branch point of $S$ and root of $T$.
	
	\item To invert the \emph{edge swap}, we need to identify the original edges $e$ and $f$. The unique edge pointing out of $b'$ towards $m$ is $f$. The original $e$ is the first edge, going back along the path
	in $T'_{b'}$ from $s(f) = b'$ to $t(f)$, that reconnects $S'_{mn,b'} \setminus f$. 
\end{itemize}

\setcounter{lemma}{0}
\begin{lemma}[``Tree surgery"]\label{treesurgery}
	Let ${\mathcal E}_{mn}$ be the set of spanning trees of $G$ containing the distinguished (undirected) edge $m \leftrightarrow n$. Then for any two distinct vertices $i, j$ of $G$,
	\begin{equation}\label{claimfrac}
	\left|\frac{\sum_{T\in  {\mathcal E}_{mn}}\sum_{S\notin  {\mathcal E}_{mn}}  w(T_i) w(S_j)}{ \sum_{T\in  {\mathcal E}_{mn}}\sum_{S\notin  {\mathcal E}_{mn}} w(T_j) w(S_i) }\right| \leq \exp(F_\mathrm{max})
	\end{equation}
	where $F_\mathrm{max}$ is the magnitude of the cycle force that is largest in magnitude, among all those associated to cycles containing the distinguished edge $m \leftrightarrow n$ (in either direction).
\end{lemma}

\begin{proof}
To prove this result, it is sufficient to find a bijection between terms in the numerator and those in the denominator, such that each term and its partner are equal or differ by a factor of $\exp(F_C)$, where $F_C$ is the cycle force associated to a cycle $C$ that contains the distinguished edge.

Consider any term $w(T_i)w(S_j)$ in the numerator. We map it to a term in the denominator as follows. Starting with the pair $(T_i, S_j)$, viewing $S_j$ as a doubly-rooted tree $S_{ij,i}$, we repeatedly apply \emph{edge swap} until the root of the rooted tree (equivalently the branch point of the doubly-rooted tree) equals $j$, unless the edge $f$ that would be removed from $T_b$ in the process is the distinguished edge ($m \to n$ or $n \to m$). In that case, apply \emph{cycle flip} in that step, so that the distinguished edge is not exchanged.

It is guaranteed that this iterative process will eventually terminate, because at every step, the branch point of the doubly-rooted tree $S_{ij,b}$ moves closer to $j$, and the part of $S_{ij,b}$ rooted at $i$ grows. Eventually, the branch point hits $j$, and the \emph{edge swap} and \emph{cycle flip} operations cannot be applied.

At the end of this iterative process the initial pair $(T_i, S_j)$ has been transformed into a pair $(T'_j, S'_i)$, whose associated weight $w(T'_j)w(S'_i)$ appears in the denominator of \eqref{claimfrac}. This defines a bijection between terms in the numerator and terms in the denominator.
 To see that the map is invertible, we note that every step along the way (an application of \emph{edge swap} or \emph{cycle flip}) is invertible, as we argued above. Therefore, as long as it is possible to uniquely determine which was applied at each step, the whole sequence of operations is invertible. But this {is} possible, because when inverting a step, we can find the edge $f$ that would have been removed from $T$ by \emph{edge swap} in that step, and that determines whether or not \emph{edge swap} or \emph{cycle flip} was {in fact} applied in that step.
 Namely, \emph{cycle flip} was applied if $f$ was the distinguished edge, and \emph{edge swap} was applied otherwise.
 
 \begin{figure}[t]
 	\centering
 	\includegraphics[width=0.5\textwidth]{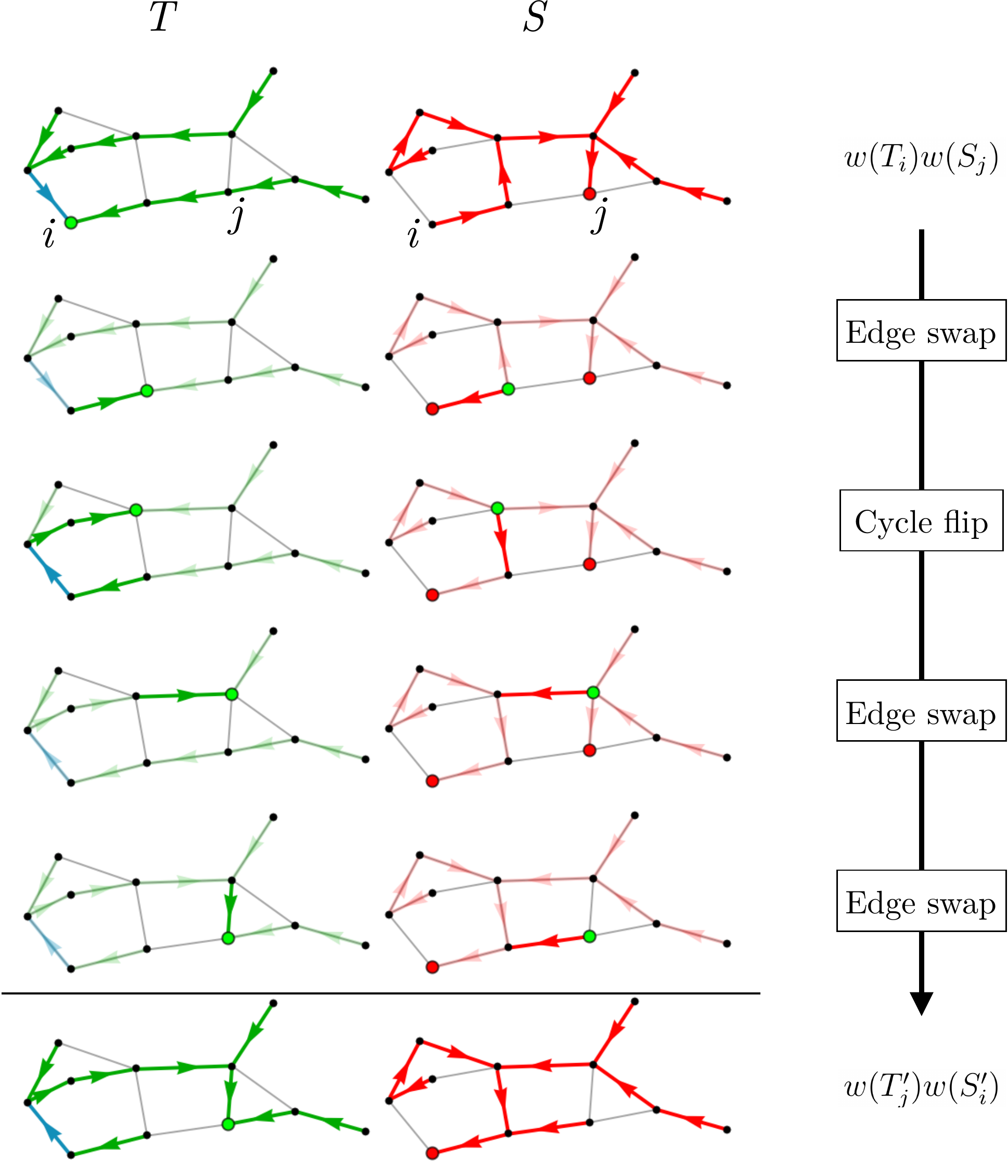}
 	\caption{{\bf Tree surgery.} Illustration of the steps of the iterative ``tree surgery" described in the proof of Lemma \ref{treesurgery} applied to particular pair of rooted spanning trees $T_i$ and $S_j$ of a graph $G$ with 11 vertices. The sequence of \emph{edge swap} and \emph{cycle flip} operations applied has the effect of swapping the roots of the trees without swapping the distinguished edge (blue). At intermediate stages, $S$ becomes a doubly-rooted tree whose branch point (labeled in green, also the root of $T$) moves between its roots $i$ and $j$ (labeled in red). The set of directed edges in the final pair of trees differs from the set in the original pair by the edges in the cycle (which contains the distinguished edge) flipped in the second step.}
 	\label{fig:tree-surgery-steps}
 \end{figure}

Having found this bijection between terms, it remains only for us to ask what is the ratio of the terms $w(T_i)w(S_j)$ and $w(T'_j)w(S'_i)$? The operation \emph{edge swap} has no effect on this product of weights, since it merely moves edges between $T$ and $S$. However, when \emph{cycle flip} is applied, edges change the way they are directed, and the weight $w(T_i)w(S_j)$ changes by a factor of $\exp(F_C)$ where $C$ is the (directed) cycle that gets flipped. Since we only apply \emph{cycle flip} if the path in $T_b$ from $t(e)$ to its root contains $m \leftrightarrow n$, $C$ is always a cycle containing $m \leftrightarrow n$. Furthermore, in the iteration described above, \emph{cycle flip} is applied {at most once}. To see this, note that the original tree $T_i$ contains either $m \to n$ or $n \to m$, never both. Furthermore, the edge $f$ that comes up in \emph{edge swap} always points from the part of $S$ rooted at $j$ to the part rooted $i$. Thus, if \emph{cycle flip} flips the distinguished edge to point the other way, it will never come up as $f$ in \emph{edge swap} again, because the part of $S$ rooted at $i$ only ever grows during this algorithm.

So we have 
\begin{equation}
\frac{w(T_i)w(S_j)}{w(T'_j)w(S'_i)} = \exp(F_C)
\end{equation}
for some cycle $C$ that contains the edge $m \leftrightarrow n$, as desired.

To prove the inequality \eqref{claimfrac}, we now match up terms above and below using this bijective ``tree surgery", putting them in an order such that each term above has the same position as its partner (e.g.~image) below. The Lemma then follows from the inequality: 
\begin{equation}
\frac{\sum_{i=1}^n x_i}{\sum_{i=1}^n y_i} = \frac{\sum_{i=1}^n \left(x_i/y_i\right)y_i}{\sum_{i=1}^n y_i} \leq \max_i \left(\frac{x_i}{y_i}\right).
\end{equation}
\end{proof}

\setcounter{lemma}{1}
\begin{lemma}[``Cycle flip only"]
\label{rootswap}
For any spanning trees $T, S$ and vertices $i, j$ of $G$
\begin{equation}
\frac{w(T_i)w(S_j)}{w(T_j)w(S_i)} \leq \exp(F_{i \leftrightarrow j})
\end{equation}
where $F_{i \leftrightarrow j}$ is {largest} possible value of $\ln [w(P_{i\to j} \cup P_{j\to i} ) / w(P_{i\to j}^* \cup P_{j\to i}^*) ]$ where $P_{i\to j}$ is a (non-self-intersecting) path from $i$ to $j$ and $P_{j\to i}$ is a  (non-self-intersecting) path from $j$ to $i$, and the superscript `$\, *$' denotes the reverse orientation.
\end{lemma}

\begin{proof}
	As above, we consider the pair $(T_i,S_j)$ but this time just apply \emph{cycle flip} to it repeatedly until it can no longer be applied (because the branch point of $S$ has become $j$). The effect of these steps is to ``swap the roots" of the two trees $T_i$ and $S_j$, changing the directions of edges without changing the underlying (undirected) spanning trees. Along any undirected spanning tree $T$, there is a unique directed path $T_{v \to w}$ from any vertex $v$ to any other vertex $w$. ``Re-rooting" a tree changes its weight as follows
	\begin{equation}
	w(T_v)w(T_{v\to w}) = w(T_w)w(T_{w\to v})
	\end{equation}
	which implies
	\begin{equation}
	w(T_i)w(S_j) = w(T_j)w(S_i)\frac{w(S_{i\to j})w(T_{j\to i})}{w(S_{j\to i})w(T_{i\to j})}.
	\end{equation}
	The fraction appearing here is of the form $w(P_{i\to j} \cup P_{j\to i} ) / w(P_{i\to j}^* \cup P_{j\to i}^*) $, as required in the statement, establishing the result.
\end{proof}

It is important to note that neither Lemma \ref{treesurgery} nor Lemma \ref{rootswap} implies the other, although their proofs can be viewed as depending on a common technique.

\section{Saturating the inequalities}\label{sec:saturate}

We have established a number of thermodynamic bounds on steady-state response to edge perturbations.
It remains an open question whether we can saturate these inequalities.  
In this section, we exhibit one example where we can saturate our bounds---the case of a system whose transition graph $G$ consists of a single cycle $C$  with cycle force $F_C=F_{\rm max}$.  
While we are unable to prove this is the only way to saturate our inequalities, we do argue for its general relevance.

To keep the discussion as straightforward and precise as possible, we focus on the ratio bound in \eqref{eq:barBound1}, as this turns out to be the simplest to investigate.
We first specialize to the case where we vary the edge parameter $B_{mn}$ associated to the edge $e_{mn}$, and ask for the response of the ratio of steady-state probabilities of the adjacent states $m$ and $n$.
In this case, the series of inequalities that lead to our bound can be summarized as
\begin{equation}\label{eq:sat}
\begin{split}
\left|\frac{\partial \ln \left(\pi_m / \pi_n\right)}{\partial B_{mn}}\right| &=\left|\frac{b_0a_1-a_0b_1}{(b_0+b_1)(a_0+a_1)}\right| \\
&\hspace{-.3cm}\underset{\textrm{AM-GM}}{\le} \tanh\left(\frac{1}{4}\left|\ln\frac{b_0 a_1}{a_0 b_1}\right|\right)\\
&\hspace{-.5cm}\underset{\textrm{``tree surgery''}}{\le} \tanh(F_C/4),
\end{split}
\end{equation}
where we use the notation
	\begin{align}
	a_1 &= \sum_{T \in {\mathcal E}_{mn}} w(T_m) \quad \quad 	a_0 = \sum_{S \notin {\mathcal E}_{mn}}  w(S_m)\\
	b_1 &= \sum_{T \in {\mathcal E}_{mn}} w(T_n) \quad \quad 	b_0 = \sum_{S \notin {\mathcal E}_{mn}}  w(S_n).
	\end{align}
The first inequality in Eq.~\eqref{eq:sat} is an application of the AM-GM inequality and the second comes about from our ``tree surgery'' argument of Lemma~\ref{treesurgery}.
We address each in turn.

Let us begin with the tree surgery inequality, which comes about from analyzing the ratio 
\begin{equation}
\frac{b_0 a_1}{a_0 b_1}=\frac{\sum_{T\in \mathcal{E}_{mn}}\sum_{S\notin \mathcal{E}_{mn}}w(T_m)w(S_n)}{\sum_{T\in \mathcal{E}_{mn}}\sum_{S\notin \mathcal{E}_{mn}}w(T_n)w(S_m)}.
\end{equation}
The tree surgery provides an invertible mapping between the terms in the numerator and those in the denominator.  
For the case of a single cycle with vertices $m, n$ adjacent to the  distinguished edge $e_{mn}$, we have 
\begin{equation}
w(T_m)w(S_n)=w(T_n)w(S_m)e^{F_C}
\end{equation}
for all $T\in{\mathcal E}_{mn}$ and $S\notin{\mathcal E}_{mn}$.
Thus, every term in the numerator is proportional to a term in the denominator with the same proportionality constant:
\begin{equation}
\frac{b_0 a_1}{a_0 b_1}=\frac{\sum_{T_m\in \mathcal{E}_{mn}}\sum_{S_n\in \mathcal{E}_{mn}}w(T_n)w(S_m)e^{F_C}}{\sum_{T_n\in \mathcal{E}_{mn}}\sum_{S_m\in \mathcal{E}_{mn}}w(T_n)w(S_m)}=e^{F_C}.
\end{equation}
Thus, the ``tree surgery'' inequality is exactly satisfied in this case.

Equality in the AM-GM inequality is reached when
\begin{equation}\label{eq:amgmSat}
a_0 b_0 = a_1 b_1.
\end{equation}
While there are numerous choices for the rates that cause this equality to be satisfied, we will just exhibit a particular one to show that it is possible.
To do so, we first make a simplifying observation: each term on both sides of the equality is a product of the weight of a spanning tree rooted at $m$ and one that is rooted at $n$.
Therefore, each term has exactly the same dependence on the vertex parameters $E_j$, so we can cancel all the $E_j$ on both sides of \eqref{eq:amgmSat}.
Thus, all we need to do is fix the symmetric and asymmetric edge parameters.
We first fix the asymmetric edge parameters by choosing all the weight of the cycle force to be on the perturbed edge $e_{mn}$,
\begin{equation}
F_{kl}=-F_{lk}=\left\{\begin{array}{cc}F_C & k=m, l=n \\0 & \mathrm{else} \end{array}\right..
\end{equation}
Solving Eq.~\eqref{eq:amgmSat} for the symmetric edge parameters then leads to the relation
\begin{equation}
e^{B_{mn}}=\sum_{\substack{e_{ij} \in G \\ ij\neq mn}}e^{B_{ij}}.
\end{equation}
Thus, it is possible to saturate our inequality for the response of the ratio $\ln(\pi_m/\pi_n)$ to perturbations of $B_{mn}$.

This may seem like a rather special case, but we believe the situation is more general than it first appears, since it is possible for the dynamics on more complicated graphs $G$ (e.g.~with multiple cycles) to effectively have this ``single-cycle" behavior.
To see this, note that if the rates of transitions in $G$ are very small, apart from those around a single cycle containing the perturbed edge $e_{mn}$, then the graph is effectively composed of a single cycle, for the purposes of understanding response of the states on the cycle.
In addition, if we look at the response of ratios of arbitrary states on the cycle, such as $\ln(\pi_i/\pi_j)$, again the dynamics can effectively reproduce the situation discussed above, where we focused on the vertices adjacent to $e_{mn}$. This is because if the rates along the unique paths from $i$ to $m$ and $j$ to $n$ on the cycle are extremely fast, the states along these paths rapidly reach a local steady-state distribution. The two paths then act as two ``effective states" adjacent to the perturbed edge $e_{mn}$.

These arguments suggest that, for a general graph $G$, there are limits of the rates that give rise to response approaching arbitrarily closely the bound set by Corollary \ref{ratiobound}.

\section{Asymmetric edge perturbation inequality}\label{sec:force}

Our asymmetric edge perturbation bound follows from a more general inequality for arbitrary perturbations of a single rate:
\begin{proposition}
\label{singlerate}
\begin{equation}
\left|\frac{\partial \pi_k}{\partial \ln W_{ij}}\right| \leq \pi_k(1-\pi_k) 
\end{equation}
\end{proposition}
\begin{proof}
By the matrix-tree theorem, we can write
\begin{equation}
\pi_k = \frac{a W_{ij} + b}{c W_{ij} + d}.
\end{equation}
where $a, b$, $c$ and $d$ are nonnegative quantities formed from sums of weights of rooted spanning trees that do not depend on $W_{ij}$. By normalization of probability $\pi_k \leq 1$, so we have $c \geq a$, $d \geq b$. Differentiating these expressions yields
\begin{equation}
\frac{\partial \pi_k}{\partial \ln W_{ij}} = \frac{(a d - b c)W_{ij}}{\left(c W_{ij} + d\right)^2}
\end{equation}
which after re-arranging gives
\begin{equation}
\begin{split}
&\left|\frac{\partial \pi_k}{\partial \ln W_{ij}}\right| \\
&\quad= \pi_k (1-\pi_k)\left|\frac{d-b}{(c-a)W_{ij}+(d-b)}-\frac{b}{a W_{ij}+b}\right|
\end{split}
\end{equation}
but the value of each of the two fractions on the right hand side is not smaller than zero or greater than one. This means their difference is not greater than one in magnitude, implying the result.
\end{proof}

\begin{corollary}
\begin{equation}
\left|\frac{\partial \pi_k}{\partial F_{ij}}\right| \leq \pi_k(1-\pi_k) 
\end{equation}
\end{corollary}
\begin{proof}
The asymmetric edge parameter $F_{ij}$ appears in two rates, $W_{ij}$ and $W_{ji}$. This implies, by the chain rule
\begin{equation}
\frac{\partial \pi_k}{\partial F_{ij}} = \frac{1}{2}\left(W_{ij}\frac{\partial \pi_k}{\partial W_{ij}} - W_{ji}\frac{\partial \pi_k}{\partial W_{ji}}\right),
\end{equation}
which implies, by the triangle inequality, 
\begin{equation}
\left|\frac{\partial \pi_k}{\partial F_{ij}}\right| \leq  \frac{1}{2}\left|W_{ij}\frac{\partial \pi_k}{\partial W_{ij}}\right| + \frac{1}{2}\left|W_{ji}\frac{\partial \pi_k}{\partial W_{ji}}\right|.
\end{equation}
Now applying Proposition \ref{singlerate} establishes the desired result.
\end{proof}


%
%
%
%
%
%
%
%
%

\section{Biochemical applications}\label{sec:biochem}

So far, we have stated and proved equalities and inequalities about the response to perturbations of physical systems whose dynamics are well-modeled as continuous in time and Markovian over a finite state space. In this section, we describe specializations of these general results to two well-known motifs found in biochemical networks. In each case, we find an inequality relating some figure of merit to a chemical potential difference driving the network out of equilibrium (for example, $\Delta\mu = \mu_\mathrm{ATP}-\mu_\mathrm{ADP}-\mu_\mathrm{Pi}$ for ATP hydrolysis).

There are several ways that studying a biochemical network might lead us to consider a linear time evolution equation like \eqref{master},
\begin{equation}
\dot{p}_i(t) = \sum_{j=1}^N W_{ij} p_j(t), 
\end{equation}
with $\sum_iW_{ij} = 0$ for all $j$. First, the chemical master equation, which governs the evolution of the distribution over counts $(n_A, n_B, \dots)$ of chemical species $A, B, \dots$, is of this form. However, for chemical systems with many particles the number of states $N$ in such a description is enormous.

However, for some chemical reaction networks, the linear equation \eqref{master} arises as the {rate equation} governing the deterministic evolution of the concentrations of chemical species. As emphasized by Gunawardena \cite{gunawardena2012linear, gunawardena2014time}, this is a generic situation that can arise from strong time-scale separation. When the rate equation of a reaction network is of the form \eqref{master}, we can equivalently view it as the master equation of a continuous-time Markov chain describing the stochastic transitions of a single molecule subject to a set of effectively monomolecular reactions \cite{mirzaev2013, Wong2017}.
Whichever interpretation we take, the mathematics that arises is the same, and our results can be put to work.

\newcommand{\ld}[1]{W_{#1}\frac{\partial}{\partial W_{#1}}}

\subsection{Covalent modification cycle}

Goldbeter and Koshland \cite{Goldbeter1981} studied a model of the covalent modification and demodification of a substrate by two enzymes, assuming the action of both enzymes obeys mass-action kinetics with a single intermediate complex and no product rebinding:
\begin{equation}\label{gkmm}
\begin{split}
E_1 + S \leftrightarrow E_1 S \to S^* + E_1 \\
E_2 + S^* \leftrightarrow E_2 S \to S + E_2 
\end{split}.
\end{equation}
The total substrate concentration $S_\mathrm{tot} = [S]+[S^*]+[E_1 S]+[E_2 S]$ is conserved in these reactions, as are the enzyme totals $E_{1,\mathrm{tot}} = [E_1] + [E_1 S]$,  $E_{2,\mathrm{tot}} = [E_2] + [E_2 S]$. In the limit of saturating substrate $S_\mathrm{tot} \gg E_{1,\mathrm{tot}}, E_{2,\mathrm{tot}}$, the kinetics are effectively Michaelis-Menten in form, and the steady-state ratio $[S^*]/[S]$ can exhibit unlimited sensitivity to changes in $E_{1,\mathrm{tot}}$ and $E_{2,\mathrm{tot}}$~\cite{Goldbeter1981}.

Sensitivity of the steady state to changes in enzyme concentrations is only possible out of equilibrium \cite{Goldbeter1987}. In \eqref{gkmm}, the nonequilibrium nature of the system is reflected in the combination
of the irreversible product release reactions with the overall reversibility of the modification of $S$. 

In the regime of low substrate  $S_\mathrm{tot} \ll E_{1,\mathrm{tot}}, E_{2,\mathrm{tot}}$, we have that $[E_1] \approx  E_{1,\mathrm{tot}}$ and $[E_2] \approx  E_{2,\mathrm{tot}}$, and the nonlinear mass-action dynamics implied by \eqref{gkmm} reduce to linear kinetics, with the enzyme concentrations ``absorbed" into the rate constants (see Fig.~\ref{fig:GK_suppl_1}). 

\begin{figure}
	\centering
	\includegraphics[scale=.35]{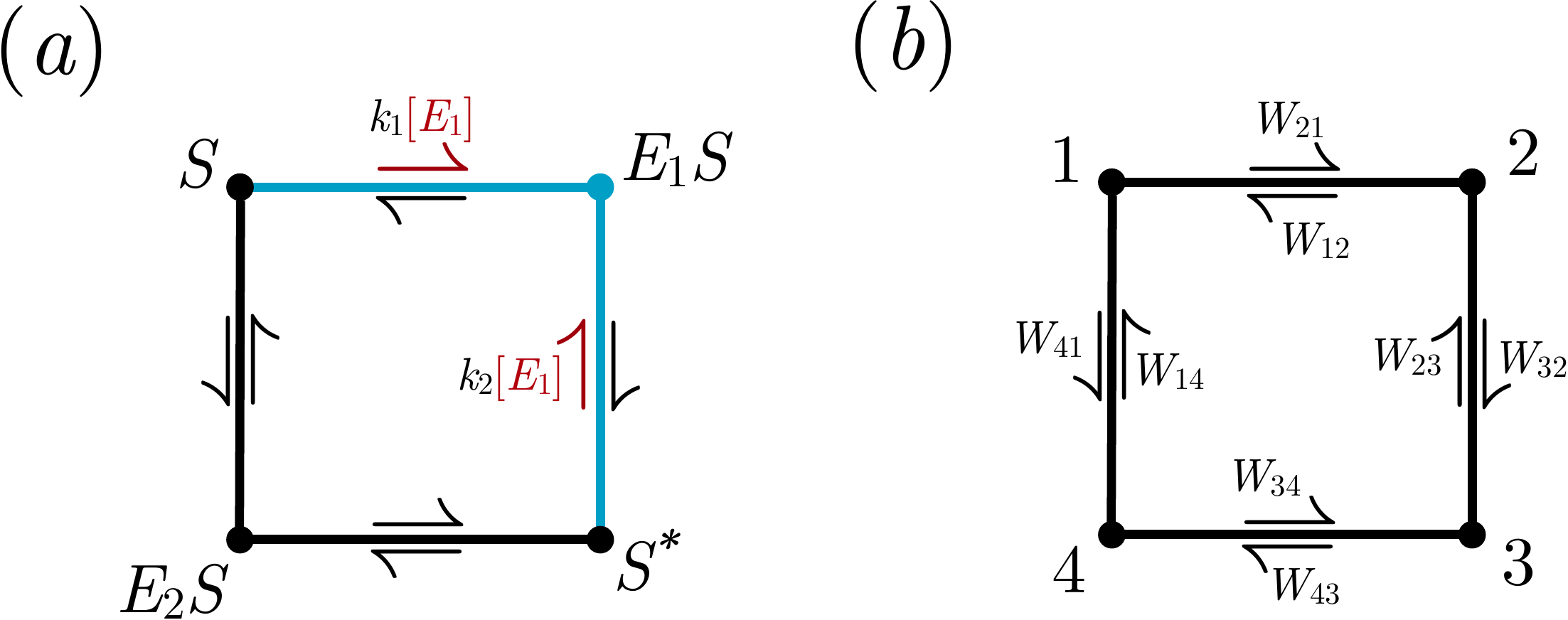}
	\caption{(a) The transition graph $G$ arising from the Goldbeter-Koshland model in the low substrate limit, with product rebinding. Two transition rates (red) depend on the (assumed constant) free enzyme concentration $[E_1]$ that we vary. Scaling $[E_1]$ is equivalent to a perturbation of two edge parameters and one vertex parameter (blue). (b) State numbers and rate labels we use in this subsection. Key equivalences are ``1 = $S$", ``3 = $S^*$", ``$W_{21} = k_1 [E_1]$", and ``$W_{23} = k_2 [E_1]$". }
	\label{fig:GK_suppl_1}
\end{figure}

In this work, we consider the low-substrate limit, and study the relative probability $\pi_{S^*}/\pi_S$ for a particular substrate molecule to be modified.
For thermodynamic consistency, all reactions must be reversible, so we must include product rebinding.
We further suppose that concentrations of other participants in these reactions  (e.g.~ATP, ADP in the case of phosphorylation/dephosphorylation) are held at fixed values. 
These choices yield a system of the form we have studied in the preceding sections, with linear dynamics of the form \eqref{master}, held out of equilibrium by the cycle force $F_C = \ln(\frac{W_{21}W_{32}W_{43}W_{14}}{W_{12}W_{23}W_{34}W_{41}})$. For a system such as this one, driven chemically, we can identify $F_C = \Delta \mu / k_\mathrm{B} T$. Our results proved above then imply a bound, in terms of $\Delta \mu$, on the sensitivity $s$ of the steady-state ratio $\pi_{S^*}/\pi_S$ to a change in $[E_1]$. 

Perturbing $[E_1]$ is equivalent to perturbing two edge parameters and a vertex parameter:
\begin{widetext}
\begin{align}
[E_1]\frac{\partial}{\partial [E_1]} &= \ld{21} + \ld{23} \\ &= \left(\ld{12} + \ld{21} \right) +\left(\ld{23} + \ld{32} \right)  -\left(\ld{12}+\ld{32}\right)\nonumber \\
&= - \frac{\partial}{\partial B_{12}} - \frac{\partial}{\partial B_{23}} - \frac{\partial}{\partial E_{2}}. \nonumber
\end{align}
\end{widetext}
Now we can apply Corollaries  \ref{ratenergypert} and \ref{rcycles} to bound the sensitivity
\begin{align}\label{sensebound}
s &= \left| [E_1] \frac{\partial \ln( \pi_3 / \pi_1)}{\partial [E_1]} \right| \\ \nonumber
&=\left| - \left(\frac{\partial}{\partial B_{12}} + \frac{\partial}{\partial B_{23}} \right) \ln \left(\frac{\pi_3}{ \pi_1}\right) - \frac{\partial \ln (\pi_3 / \pi_1)}{\partial E_{2}} \right| \\&\leq \tanh(F_{1 \leftrightarrow 3}/4) = \tanh(\Delta \mu/4k_\text{B}T).  \nonumber
\end{align}

Remarkably, the form of the bound \eqref{sensebound} remains unchanged even if the assumption that catalysis proceeds via a single intermediate complex is completely relaxed.  In particular, following Gunawardena et al.~\cite{xu2012realistic, dasgupta2014fundamental}, we consider an arbitrary reaction network built out of a collection of any number of reactions of the following form, which include an arbitrary number of intermediates and reactions between them:
\begin{equation}
\begin{split}
E_1 + S \leftrightarrow (E_1 S)_i \\ 
E_1 + S^* \leftrightarrow (E_1 S)_i \\ 
(E_1 S)_i \leftrightarrow (E_1 S)_j \\ 
E_2 + S \leftrightarrow (E_2 S)_i \\ 
E_2 + S^* \leftrightarrow (E_2 S)_i \\ 
(E_2 S)_i \leftrightarrow (E_2 S)_j.
\end{split}
\end{equation}
A general network of this form is schematically represented in Fig.~\ref{fig:phospho}(b). In any such network, consider the subgraph whose vertices $V$ are all the intermediates $\{(E_1 S)_i\}$ containing $E_1$, together with $S$ and $S^*$, and whose edges $\mathcal{E}$ are all the edges between the vertices $V$. Scaling $[E_1]$ is equivalent to decreasing all the edge parameters associated to edges in $\mathcal{E}$, and the vertex parameters associated to vertices in the set $V_I = V\setminus\{S,S^*\}$. This decomposition yields the result
\begin{align}\label{generalGKbound}
s &= \left| [E_1] \frac{\partial \ln( \pi_{S^*} / \pi_{S})}{\partial [E_1]} \right| \\ \nonumber
&= \left| -\left(\sum_{v \in {V_I}} \frac{\partial}{\partial E_v} - \sum_{e \in {\mathcal{E}}} \frac{\partial}{\partial B_e} \right)\ln\left(\frac{\pi_{S^*}}{ \pi_S}\right)\right| \\ &\leq \left|\left(\sum_{e \in {\mathcal{E}}} \frac{\partial}{\partial B_e} \right)\ln \left(\frac{\pi_{S^*}}{ \pi_S}\right)\right|\nonumber \\ & \leq \tanh(F_{S^* \leftrightarrow S}/4) = \tanh(\Delta \mu / 4 k_\mathrm{B} T) \nonumber
\end{align}
where the last line follows from Corollary \ref{forGK} with $W = \{ S, S^* \}$, $|W| = 2$.

\subsection{Biochemical sensing and the Govern--ten Wolde trade-off}
\label{tenWolde_consistency}

Now we will show how our sensing bound \eqref{eq:sensingbound} arises, and how it reduces to the results of Govern and ten Wolde \cite{govern2014PNAS} in the appropriate limits. 

To arrive at \eqref{eq:sensingbound}, we first employ the approximation \eqref{eq:errorprop}
\begin{align}
\label{eq:sensingboundderiv}
\left(\frac{\sigma_e}{[E_1]}\right)^2 &\approx \frac{\sigma_{s}^2}{(\partial \mu_s/\partial [E_1])^2} \frac{1}{[E_1]^2} \nonumber\\ &= \frac{N \pi_{S^*}(1-\pi_{S^*})}{N^2([E_1]\partial \pi_{S^*}/\partial [E_1])^2}.
\end{align}
Now we recognize the derivative in the denominator as being an instance of the kinds we have considered already. In particular, application of \eqref{eq:barrierBoundMulti}, in the form of a ratio of general observables, together with our vertex perturbation results, yields
\begin{equation}
\left|[E_1]\frac{\partial \pi_{S^*}}{\partial [E_1]} + \pi_{S^*} \pi_{Y}\right| \leq \pi_{S^*} (1-\pi_{S^*})\tanh(\Delta \mu/4 k_\mathrm{B} T)
\end{equation}
where $\pi_Y$ is the fraction of the total substrate bound up in complexes involving $E_1$. In \cite{govern2014PNAS}, these enzymatic intermediates are neglected (e.g. equation [S22] of \cite{govern2014PNAS}), and to proceed we will do the same here, supposing that $\pi_Y$ is very small. We then get
\begin{align} 
\left(\frac{\sigma_e}{[E_1]}\right)^2 &\gtrsim \frac{1}{N \pi_{S^*} (1-\pi_{S^*})  \tanh(\Delta \mu/4 k_\mathrm{B} T)^2} \nonumber \\  &\geq \frac{4}{N \tanh(\Delta \mu/4 k_\mathrm{B} T)^2}
\end{align}
as desired. In the limit $\Delta \mu \gg k_\mathrm{B}T$,
\begin{equation}
\left(\frac{\sigma_e}{[E_1]}\right)^2 \gtrsim \frac{4}{N}
\end{equation}
whereas in the limit $\Delta \mu \ll k_\mathrm{B}T$, this yields
\begin{equation}\label{lowforcelimit}
\left(\frac{\sigma_e}{[E_1]}\right)^2 \gtrsim \frac{64}{N}\left(\frac{ k_\mathrm{B} T}{\Delta \mu}\right)^2.
\end{equation}

To make contact between our low potential limit $\Delta \mu \ll k_\mathrm{B}T$ \eqref{lowforcelimit} and the results of Govern and ten Wolde, we now review the context of their results. In that paper, the authors study the error $\delta c/c$ in an estimate of the concentration $c$ of an extracellular ligand $L$, which binds to receptor $R$, forming a complex: $R + L \leftrightarrow RL$. The complex $RL$ then plays the role of $E_1$ in our discussion above, participating in a covalent modification cycle. The concentration of the ligand-bound receptor (in our notation $E_1$, in their notation $RL$) is given by 
\begin{align}\label{fractionbound}
[E_1] = [RL] &= R_T p \nonumber \\
p &= \frac{c}{c+K} 
\end{align}
where $K$ is the dissociation constant and $R_T$ is the total concentration of receptors. 

An estimate of $c$ can be constructed, just as we describe in the main text producing an estimate of $[E_1]$. Using the approximation \eqref{eq:errorprop} together with the equation \eqref{fractionbound} relates the error of these estimates,
\begin{equation}\label{error_related}
\left(\frac{\delta c}{c}\right)^2 \approx  \frac{1}{(1-p)^2}\left(\frac{\sigma_e}{[E_1]}\right)^2
\end{equation}
The authors of \cite{govern2014PNAS} compare the sensing error $(\delta c/c)^2$ not to the chemical potential difference $\Delta \mu$, but to a quantity $w$, which is the dissipation rate of the whole system, normalized by the sum of all the rates (forward and reverse) around the modification cycle (in the case studied by the authors, consisting of only two states, in our notation, $S$ and $S^*$, and no intermediates). We shall call this quantity, which is the rate of relaxation to steady state, $\mathcal{R}$. So in our notation,  
\begin{equation}\label{w_def}
w = \left(N \times |J| \times \frac{\Delta \mu}{k_\mathrm{B} T}\right)/\mathcal{R}. 
\end{equation}
where $J$ is the net current around the cycle, per substrate molecule. The arguments of Govern and ten Wolde show that in the limit $\Delta \mu \ll k_\mathrm{B} T$, 
\begin{equation}\label{tenWolde_error}
\left(\frac{\delta c}{c}\right)^2 \geq \frac{4}{(1-p)^2 w}.
\end{equation}
This is slightly tighter than the inequality (S26) they write in \cite{govern2014PNAS}, but also follows from their argument. As a consequence of \eqref{error_related} and \eqref{tenWolde_error}, we then get
\begin{equation}\label{tenWolde_result}
\left(\sigma_e/[E_1]\right)^2 \geq 4/w.
\end{equation}
To show that this bound coincides with our result \eqref{lowforcelimit}, we need to show that for fixed small force the smallest value achieved by $4/w$ is $64 (k_\mathrm{B}T)^2/ N (\Delta \mu)^2$. If this were not so, it would imply that either our result was weaker than that of Govern and ten Wolde, or vice versa.

To show that \eqref{lowforcelimit} and \eqref{tenWolde_result} do coincide for fixed small force, we use the inequality,
\begin{equation}\label{jrineq}
\frac{|J|}{\mathcal{R}} \leq \pi_{S^*}(1-\pi_{S^*})\tanh\left(\frac{\Delta \mu}{4 k_\mathrm{B}T}\right) \leq  \frac{\Delta \mu}{16 k_\mathrm{B}T}
\end{equation}
which holds (for a two-state system), by the same algebra \eqref{main1}--\eqref{main3} that led to our symmetric edge perturbation result. See also Malaguti and ten Wolde \cite{malaguti2019theory} (equations S112 to S114), who give an explicit expression for $|J|/R$.

Plugging \eqref{jrineq} into the definition \eqref{w_def}, we get
\begin{equation}\label{result_rel}
\frac{4}{w} \geq \frac{64}{N}\left(\frac{ k_\mathrm{B} T}{\Delta \mu}\right)^2.
\end{equation}
Additionally, this inequality can be saturated in the limit $\Delta \mu \ll k_\mathrm{B}T$, because there is a near equilibrium regime in which $|J|/R \approx \Delta \mu / 16 k_\mathrm{B}T$. This establishes the desired equivalence between our results.

\subsection{Kinetic proofreading}

In our presentation and analysis here, we follow closely the  papers of Murugan et al. \cite{murugan2012speed, Murugan2014}. Our results generalize bounds on the discriminatory index $\nu$  found in those works. 

First, we consider the single-loop, three-state network (see Fig.~\ref{fig:KP_suppl_1}) equivalent to the system studied by Hopfield and Ninio \cite{Hopfield1974, Ninio1975}.
\begin{figure}[t]
	\centering
	\vspace*{0.3 cm}
	\includegraphics[scale=.35]{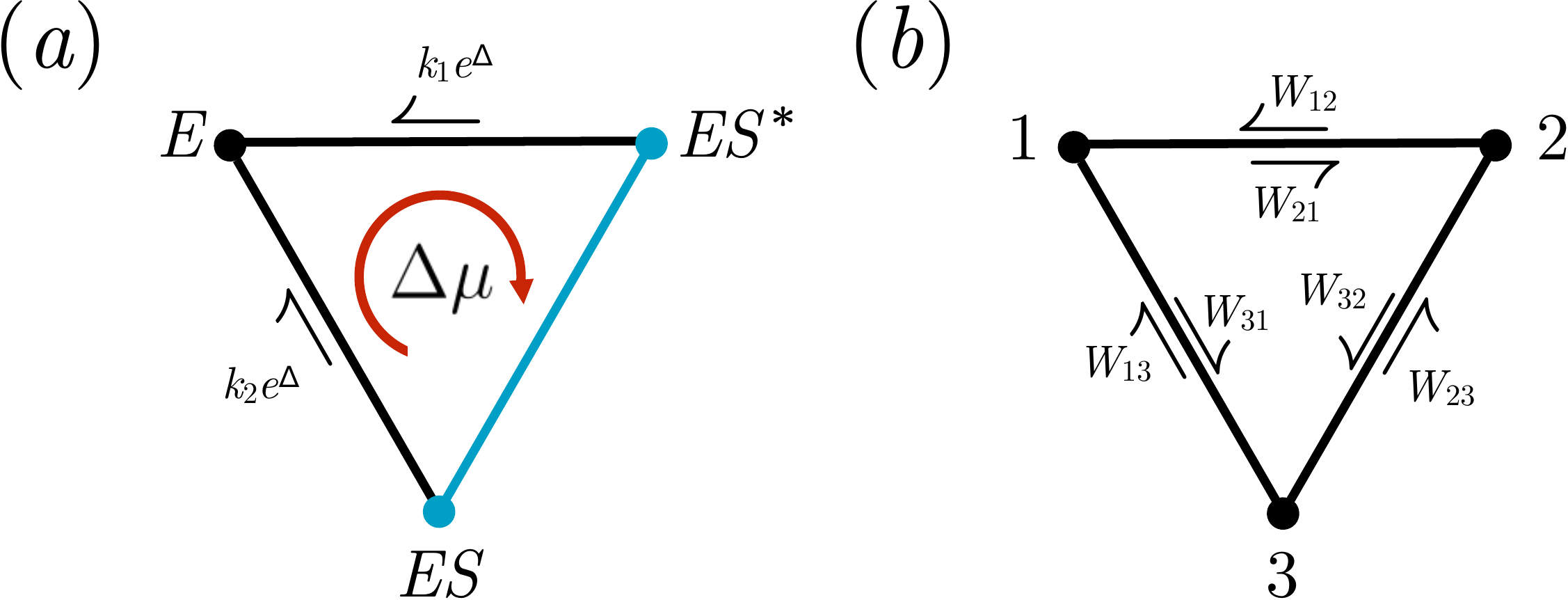}
	\caption{(a) The transition graph $G$ for the single-loop kinetic proofreading mechanism of Hopfield. The dissociation rates of the complexes $ES^*$ and $ES$ are the only rates that depend on the binding energy $\Delta$. Perturbing $\Delta$ is equivalent to vertex and edge perturbations (blue). (b) State numbers and rate labels we use in this subsection. Key equivalences are ``1 = $E$", ``3 = $ES$", ``$W_{12} = k_1 e^\Delta$", and ``$W_{13} = k_2 e^\Delta$".}
	\label{fig:KP_suppl_1}
\end{figure}
A perturbation of the binding energy $\Delta$ can be decomposed as a linear combination of vertex and symmetric edge parameter perturbations. In terms of the notation we introduce in Fig.~\ref{fig:KP_suppl_1}(b), we have
\begin{widetext}
\begin{align}
\frac{\partial}{\partial \Delta} &= \ld{12} + \ld{13} \\ &= \left( \ld{13}+\ld{23}\right)+\left(\ld{12}+\ld{32}\right)-\left(\ld{23}+\ld{32}\right) \nonumber \\
&= \frac{\partial}{\partial E_3} + \frac{\partial}{\partial E_{2}} + \frac{\partial}{\partial B_{23}} .\nonumber
\end{align}
Now we can apply Corollaries \ref{ratenergypert} and \ref{ratiobound} to derive the bound
\begin{align}
\left|\nu - 1\right| &= \left|\frac{\partial \ln (\pi_1/\pi_{3})}{\partial \Delta} - 1\right| = \left|\frac{\partial \ln (\pi_1/\pi_{3})}{\partial E_3} + \frac{\partial \ln (\pi_1/\pi_{3})}{\partial E_{2}} + \frac{\partial \ln (\pi_1/\pi_{3})}{\partial B_{23}}- 1\right| \\ \nonumber
&\leq |1 + 0 - 1| + \left|\frac{\partial \ln (\pi_1/\pi_{3})}{\partial B_{23}}\right| \\ \nonumber 
& \leq  \tanh\left(F_\mathrm{max}/4\right) =  \tanh\left(\Delta\mu/4k_\mathrm{B} T\right).\nonumber
\end{align}
\end{widetext}

In the case of the more general kinetic proofreading scheme \cite{murugan2012speed, Murugan2014} where $m$ complexes can dissociate, described in Fig.~\ref{fig:KP_suppl_2}, perturbing $\Delta$ is equivalent to perturbing the edge and vertex parameters associated to the edges $\mathcal{E}$ and vertices $V$ on one side of the ``fence". We then have
\begin{align}
\left|\nu - 1\right| &= \left|\frac{\partial \ln (\pi_E/\pi_{ES})}{\partial \Delta} - 1\right| \\ \nonumber
&=\left| \left(\sum_{v \in {V}} \frac{\partial}{\partial E_v} + \sum_{e \in {\mathcal{E}}} \frac{\partial}{\partial B_e} \right)\ln\left(\frac{\pi_{E}}{ \pi_{ES}}\right)-1\right|  \\ \nonumber
&\leq |1 - 1| + \left|\left(\sum_{e \in {\mathcal{E}}} \frac{\partial}{\partial B_e} \right)\ln \left(\frac{\pi_{E}}{ \pi_{ES}}\right)\right| \\
&\leq (m-1)\tanh(F_{E \leftrightarrow ES}/4)\nonumber
\end{align}
where in the last line we have applied Corollary \ref{forGK}.


%

\end{document}